\documentclass[11pt]{article}
\usepackage[left=1in, top=1in, right=1in, bottom=1in]{geometry}

\usepackage{microtype}
\usepackage{graphicx}
\usepackage{subfigure}
\usepackage{booktabs} 
\usepackage{amsmath, amsthm, amssymb, bm}
\usepackage{mathtools}
\usepackage{thm-restate}
\usepackage[authoryear, round]{natbib}
\usepackage{color}

\usepackage{hyperref}
\hypersetup{
    colorlinks=false,
    linkcolor=blue,
    filecolor=magenta,      
    urlcolor=cyan,
}

\newcommand{\AutoAdjust}[3]{\mathchoice{ \left #1 #2  \right #3}{#1 #2 #3}{#1 #2 #3}{#1 #2 #3} }
\newcommand{\Xcomment}[1]{{}}

\newcommand{\InBrackets}[1]{\AutoAdjust{[}{#1}{]}}
\newcommand{\Ex}[2][]{\operatorname{\mathbf E}_{#1}\InBrackets{#2}}
\newcommand{\Prx}[2][]{\operatorname{\mathbf{Pr}}_{#1}\InBrackets{#2}}

\newcommand{\dd}{\mathrm{d}}  
\newcommand{\given}{\;\mid\;}
\newcommand{\eps}{\epsilon}

\newtheorem{thm}{Theorem}[section]

\newtheorem{lemma}[thm]{Lemma}
\newtheorem{defn}[thm]{Definition}
\newtheorem{claim}[thm]{Claim}
\newtheorem{corollary}[thm]{Corollary}

\DeclareMathOperator*{\argmax}{arg\,max}

\newcommand{\noaccents}[1]{#1}
\newcommand{\newagentvar}[3][\noaccents]{%
\expandafter\newcommand\expandafter{\csname #2\endcsname}{#1{#3}}%
\expandafter\newcommand\expandafter{\csname #2s\endcsname}{#1{\boldsymbol{#3}}}%
\expandafter\newcommand\expandafter{\csname #2smi\endcsname}[1][i]{#1{\boldsymbol{#3}}_{-##1}}%
\expandafter\newcommand\expandafter{\csname #2i\endcsname}[1][i]{#1{#3}_{##1}}%
\expandafter\newcommand\expandafter{\csname #2ith\endcsname}[1][i]{#1{#3}_{(##1)}}%
}

\newcommand{\newvecagentvar}[3][\noaccents]{%
\expandafter\newcommand\expandafter{\csname #2\endcsname}{#1{\boldsymbol{#3}}}%
\expandafter\newcommand\expandafter{\csname #2s\endcsname}{#1{\boldsymbol{#3}}}%
\expandafter\newcommand\expandafter{\csname #2smi\endcsname}[1][i]{#1{\boldsymbol{#3}}_{-##1}}%
\expandafter\newcommand\expandafter{\csname #2i\endcsname}[1][i]{#1{\boldsymbol{#3}}_{##1}}%
\expandafter\newcommand\expandafter{\csname #2ith\endcsname}[1][i]{#1{#3}_{(##1)}}%
}

\newcommand{\bluecom}[1]{\textcolor{blue}{[#1]}}

\newagentvar{type}{t}
\newagentvar{alloc}{x}
\newagentvar{pay}{p}
\newagentvar{typespace}{T}
\newagentvar{val}{v}
\newagentvar{util}{u}
\newagentvar{payment}{p}
\newagentvar{dist}{F}
\newagentvar{bid}{b}

\newagentvar{expostU}{U}
\newagentvar{empDist}{E}

\newcommand{\interimalloc}{x}
\newcommand{\interimpay}{p}

\newcommand{\nsample}{m}
\newagentvar{sample}{s}
\newcommand{\alg}{\mathcal{A}}
\newcommand{\strategyset}{\mathcal{B}}

\newcommand{\posdist}{\dist^+}
\newcommand{\negdist}{\dist^-}

\newcommand{\distconst}{c_1}
\newcommand{\posbid}{\bid^+}
\newcommand{\negbid}{\bid^-}
\DeclareMathOperator{\kl}{KL}

\newcommand{\diffalg}{\mathcal{H}}
\newcommand{\diffconst}{c_2}
\newcommand{\diffset}{\theta}
\newcommand{\coordsets}{\mathcal {S}}

\DeclareMathOperator{\emp}{Emp}
\DeclareMathOperator{\empp}{Empp}
\newcommand{\func}{h}
\newcommand{\funcClass}{\mathcal{H}}
\newcommand{\PinputSpace}{\mathcal{X}}
\newagentvar{Pinput}{x}
\newagentvar{Plabel}{l}
\newagentvar{Pwitness}{t}
\DeclareMathOperator{\sgn}{sgn}
\DeclareMathOperator{\Pdim}{Pdim}
\newagentvar{permSample}{{\tilde{s}}}

\newagentvar{cost}{c}
\newcommand{\highval}{H}

\newagentvar{reserve}{r}
\newcommand{\hypoclass}{\funcClass}
\newcommand{\PDim}{\Pdim}

\newagentvar{threshold}{r}
\newcommand{\DtF}{\mu}
\newcommand{\FtD}{\lambda}
\newagentvar{fstrat}{\beta}
\newagentvar{dstrat}{\alpha}
\newagentvar{dtime}{t}
\newagentvar{dbid}{b}
\DeclareMathOperator{\DA}{DA}
\DeclareMathOperator{\FPA}{FPA}

\DeclareMathOperator{\SW}{SW}
\DeclareMathOperator{\OPT}{OPT}

\newcommand{\varalloc}{{\mathbb{A}}}
\newcommand{\varinspect}{{\mathbb{I}}}

\newcommand{\old}[1]{}

\title{Learning Utilities and Equilibria in Non-Truthful Auctions\footnote{
A previous version of this paper has been accepted to NeurIPS 2020.
We would like to thank Anna Karlin, Zhe Feng, Zihe Wang, Weiran Shen for helpful discussions.
The work was funded by an NSERC Discovery Grant, NSERC Discovery Acceleration Grant, and a Canadian Research Chair stipend.}}

\author{Hu Fu\thanks{University of British Columbia. \href{mailto:hufu@cs.ubc.ca}{\texttt{hufu@cs.ubc.ca}}}
\and Tao Lin\thanks{Harvard University. \href{mailto:tlin@g.harvard.edu}{\texttt{tlin@g.harvard.edu}}. The work was done when Tao Lin was at the Center on Frontiers of Computing Studies (CFCS), Department of Computer Science, Peking University, and during a visit to the University of British Columbia.}
}

\date{October, 2022}

\begin{document}
\maketitle

\begin{abstract}
In non-truthful auctions, agents' utility for a strategy depends on the strategies of the opponents and also the prior distribution over their private types; the set of Bayes Nash equilibria generally has an intricate dependence on the prior.  Using the First Price Auction as our main demonstrating example, we show that $\tilde O(n / \epsilon^2)$ samples from the prior with $n$ agents suffice for an algorithm to learn the interim utilities for all monotone bidding strategies.  As a consequence, this number of samples suffice for learning all approximate equilibria.  We give almost matching (up to polylog factors) lower bound on the sample complexity for learning utilities.  We also consider a setting where agents must pay a search cost to discover their own types.  Drawing on a connection between this setting and the first price auction, discovered recently by Kleinberg et al.~(2016), we show that $\tilde O(n / \epsilon^2)$ samples suffice for utilities and equilibria to be estimated in a near welfare-optimal descending auction in this setting.  En route, we improve the sample complexity bound, recently obtained by Guo et al.~(2021), for the Pandora's Box problem, which is a classical model for sequential consumer search.
\end{abstract}

\section{Introduction}
\label{sec:intro}

Mechanism design devises systems in which multiple agents take strategic actions based on their private preferences (designated as \emph{types}).  
For example, an auctioneer devises rules that determine an auction's winner and payments, based on bidders' actions (the bids); the bidders then, knowing the rule and their types --- in this case their own values for the item at sale --- strategize over their bids.
The following task is central to many aspects of mechanism design: given agents' strategies, evaluate each agent's performance, or \emph{utility}.
To start with, agents are most often interested in predicting the performance of their strategies given what the other agents do; nowadays, auctioneers and third-party service providers often give guidance to bidding, and are interested in such evaluations as well.
An auctioneer often would like to find out if a profile of strategies best respond to each other and are hence \emph{at equilibrium}; revenue, welfare and surplus analysis at equilibrium is all based on utility estimation.
Recent development in online ad auctions (such as the oCPX auctions) sees growing popularity of delegated bidding, where bidders entrust the auctioneer/platform with the task of bidding.  Auctioneers in this scenario must estimate the bidders' utilities given their bidding strategies.
As another application of utility estimation, one may be interested in knowing how far a mechanism is from being \emph{incentive compatible}, i.e., supposing all other agents bid their true values, how much incentive an agent would have to lie about her value \citep[see, e.g.,][]{BSV19}.



Since a bidder's individual type is known only to herself but may greatly affect her actions and the payoff of everyone else --- a scenario known as an \emph{incomplete information game} --- more is needed for utility prediction.
The canonical paradigm assumes that each agent's type is drawn from a distribution \citep{Harsanyi67}.
The utility prediction problem becomes, supposing agents other than~$i$ use given strategies and their types are drawn from the distributions, what is agent~$i$'s utility \emph{in expectation} for using a certain strategy?  
This so-called \emph{interim} utility is at heart of much modern mechanism design literature from both economics and computer science  \citep[e.g.][]{Mye81, LOV95, ST13, HL15}.


As is natural for many other settings with distributional prior knowledge, 
it is important to know how sensitive the problem is 
to the accuracy of the prior distributions; in particular, given sample access to the distribution, how many samples are needed to estimate the interim utilities of agents, for any profile of strategies?
This is the question we study in this work.


As we discuss below, without appealing to the structure of a game, it is hopeless to learn all utilities with a finite number of samples.  
In this work we take the First Price Auction (FPA) and the All Pay auction as two examples of non-truthful auction, and show that a small number of samples from the type distributions suffice to learn interim utilities for all \emph{monotone} bidding strategies.

The first price auction is one of the most ubiquitous and fundamental auctions, with a large literature on its various aspects, and is gaining more popularity recently for various practical reasons \citep[see e.g.][]{fpa1, fpa2, AL18}.
In the FPA, the item is sold to the bidder with the highest bid, who then pays her own bid; all other bidders pay nothing.
Our main result is that, for an FPA with $n$ bidders, $\tilde O(n / \eps^2)$ samples from the type distribution suffice to learn, with high probability, the interim utilities up to $\eps$ additive error for any profile of monotone bidding strategies.\footnote{The $\tilde O(\cdot)$ notation omits polylogarithmic factors.}  
As a corollary, this number of samples also suffice to learn the set of all $\eps$-Bayes Nash equilibria of an FPA, and to compute an $\eps$-Bayes equilibrium in polynomial time.
We also show that our bound is tight up to a logarithmic factor.




\paragraph{Sampling from Type Distributions.}
The assumption of having sample access to the underlying type distribution is standard in a long line of work on learning in mechanism design \citep[e.g.][]{CR14, MR15, MR16, BSV16, BSV18, GN17, Syrgkanis17, GW18, guo2019settling, brustle_multi-item_2020, yang_learning_2021, guo_generalizing_2021}.  
In particular, type samples have been assumed in the context of learning in non-truthful auctions \citep{BSV19, vitercik2021automated}.  
Just as in classical microeconomics, prior knowledge (in the form of samples here) comes from market research, survey, simulation etc., and is not assumed to be from past bidding history.  
We distance our approach from the line of work on learning non-truthful auctions where samples are from past bidding history \citep{CHN17, HT19}.  
This latter approach, with obvious merits, has its limitations.  
Crucially, it assumes that the observed bidding in a non-truthful mechanism is at equilibrium, which may not be the case in reality; also, to avoid strategic issues between auctions, the bidders need to be short lived or myopic.  
The two approaches (type samples vs.\@ bid samples) complement each other even in learning non-truthful mechanisms; this work takes the first approach, and leaves the direction with bid samples as an enticing open question.  

\paragraph{Our Techniques.}
Let us draw an analogy with classical PAC learning, where one typically has a class of functions on an input space and an unknown distribution over the input; with samples from the input distribution, one is to learn the expectation of the functions on the unknown distribution.
In our setting, every profile of strategies, one for each agent, defines a function which maps a profile of types to a utility for each agent; given the underlying value distribution, the expectation of this function is the agents' interim utilities.
We would like to have a number of samples such that, with high probability, the empirical estimation of utilities for \emph{all} bidding strategies is approximately correct.



With any finite number of samples, 
it is impossible to learn utilities for all arbitrary bidding strategies. 
Consider the following example: in an auction with two bidders, the first bidder having value~$1$ and bidding~$\frac 1 2$, and the second bidder's value~$\vali[2]$ uniformly drawn from $[0, 1]$; 
any finite set of samples of~$\vali[2]$ has probability measure~$0$ in the distribution of~$\vali[2]$; 
therefore on any set of samples, there are bidding strategies of bidder~$2$ that look the same on the sampled values but give bidder~$1$ drastically different utilities in expectation on the value distribution.

To overcome this, we observe that non-monotone bidding strategies are (weakly) dominated by monotone ones in a first price auction ---  a bidding strategy is monotone if its bid is non-decreasing with the value.  
For all practical purposes, therefore, learning utilities for monotone bidding strategies suffices.
A key technical insight of this work is that the restriction to monotone strategies makes it possible to learn utilities with a small number of samples.


Our sample complexity bound comes from bounding the \emph{pseudo-dimension} of the class of utility functions for monotone bidding strategies.
Pseudo-dimension is a generalization of the VC-dimension, and has been applied to learning in mechanism design \citep{MR15, MR16, BSV18, BSV19}.
Another key step is a lemma (Lemma~\ref{lem:relation-uniform-convergence}), highly related to a concentration inequality shown by \citet{DHP16}, that translates learning on a correlated distribution from samples to learning on a natural empirical product distribution. 


We extend our main result to the \emph{all pay auction}.  
In the all pay auction, the item is sold to the highest bidder, but all bidders, including the losers, pay their bids.  
This auction is a good model for crowdsourcing \citep{CHS12}.

\paragraph{Lower Bound.} We give a nontrivial, information-theoretic lower bound for the utility learning problem, showing that $\Omega(n/ \eps^2)$ samples are needed for learning utilities with monotone bidding strategies.  
The polynomial dependence on~$n$, the number of bidders, is interesting.
In particular, this highlights the difference between learning utilities for a fixed profile of bidding strategies and our task, which is to simultaneously learn the utilities for all strategies.  
For the first task, one may simply estimate, for each bidder, the cumulative density function of the highest bid among the opponents; it takes only $\Theta(1 / \eps^2)$ samples for the estimated CDF to be accurate everywhere.

\paragraph{Auctions with Search Costs.}
We extend our sample complexity bound to auctions with \emph{search costs}. 
In this setting, bidders know their value distributions but need to incur a cost to learn their true values.
Real estate and insurance markets exhibit this feature.
\citet{KWW16} showed that the Dutch auction beats the English auction and is near optimal at equilibrium for welfare in this setting.
%
%
If the value distributions are unknown, how many samples are needed for equilibria of a Dutch auction with search costs to be learned? 
A real estate market maker may be interested in this: he may collect market information and suggest search and bidding strategies for the participants.  
We show that $\tilde O(n / \eps^2)$ samples again suffice.
Specifically, with high probability we can reduce a Dutch auction with search costs on the true distribution to a first price auction without search costs on a transformed empirical distribution consisting of $\tilde O(n / \eps^2)$ value samples; any equilibrium in the first price auction can be transformed to an approximate equilibrium in the Dutch auction whose welfare is at least $(1-1/e)$-fraction of the maximal achievable welfare, up to an additive error of~$\eps$.

The maximal achievable welfare is obtained by Weitzman's algorithm for the \emph{Pandora's Box} problem  \citep{Weitzman79}, a classical model for sequential costly search.
En route, we improve a sample complexity bound for Pandora's Box problem recently obtained by \citet{guo_generalizing_2021}.  
%




\paragraph{Additional Related Works.}
\label{sec:related}
Most works cited above on learning in mechanisms are on learning revenue optimal mechanisms.  The following two exceptions are particularly close to our work.
Given a non-truthful auction, \citet{BSV19} studied the number of value samples needed to learn the maximal interim utility a bidder with some value could gain by non-truthful bidding, when all other bidders are truthful.  
In comparison, we learn interim utilities when all bidders use arbitrary monotone bidding strategies; this suffices for the study of virtually all properties of an auction, including the task of \citet{BSV19}
\footnote{
Our results imply that the maximal interim utility (w.r.t~the opponents' value distribution) obtained by non-truthful bidding can be approximated by the maximal obtainable interim utility w.r.t.~the empirical distribution, which can be computed by enumerating the samples in the empirical distribution because a best-responding bid must be equal to (or slightly more than) some opponent's value from the empirical distribution. 
}.
Our results rely on the monotonicity of the bidding strategy, which is natural in single item auctions; whether extension is possible to multi-parameter settings in a similar manner is an interesting question.

\citet{AGCU19} studied learning utilities in simulation-based games. 
They considered abstract conditional normal-form games, where the conditions play the role of types in auctions.
They used Rademacher complexity, similar to pseudo-dimension, to bound the number of sample conditions, but did not give such bounds for concrete games.
Bounding these complexity measures is generally challenging if not impossible.  
For example, without monotonicity in bidding strategies, the pseudo-dimension of utilities even in first price auctions is unbounded, as we discussed above.

There is a large literature on the computation of equilibrium in the first price auction \citep[e.g.][]{Marshall94, Gayle08, EMS09, filos-ratsikas_complexity_2021}.  
There has been major recent progress on the problem for discrete distributions \citep{SWZ20}.  
Our work provides additional motivation to study the problem on discrete distributions.

\citet{guo_generalizing_2021} gave a general method to bound sample complexity for learning over product distributions.  
Their main result requires a ``strong monotonicity'' property on the problem, which is satisfied by the Pandora Box problem.
Utilities for monotone bidding strategies, however, do not satisfy their property; our bound for the Pandora Box problem also slightly improves over theirs.

\section{Preliminaries on Auctions}
\label{sec:prelim}
\paragraph{Auctions.}
In a single item auction, $n$ bidders compete for the item.
We use $[n]$ to denote $\{1, \ldots, n\}$.
Each bidder~$i\in[n]$ has a private value $\vali$ drawn from a distribution $\disti$ supported on $\typespacei \subseteq \mathbb R_+$.
$|\typespacei|$ can be infinite. 
Distributions of different bidders are independent and can be non-identical.
In the auction, each bidder~$i$ makes a sealed-envelope bid of~$\bidi \geq 0$, and the auction maps the vector of bids~$\bids$ to an \emph{allocation} and \emph{payments}. 
$\alloci(\bids) \in [0, 1]$ denotes the probability with which bidder~$i$ receives the item at bid vector~$\bids$, with $\sum_{i=1}^n \alloci(\bids) \leq 1$, and $\payi(\bids)$ is the (expected) payment made by bidder~$i$.
In the \emph{first price auction} (FPA), the highest bidder (assuming there is no tie) wins the item and pays her bid, and no other bidder pays anything.  In the \emph{all pay auction}, the highest bidder wins the item but every bidder pays her bid.

Bidder~$i$'s \emph{ex post utility} is
$
\expostUi(\vali, \bidi, \bidsmi)\coloneqq \vali \alloci(\bids) - \payi(\bids).
$

As is standard in auction theory, we use $\valsmi$ to denote the vector $(\vali[1], \ldots, \vali[i-1], \vali[i+1], \ldots, \vali[n])$.  
Other vectors with subscript ``$-i$'' are similarly defined.

\paragraph{Strategies and Equilibrium.}
A \emph{bidding strategy} maps a bidder's value to a bid.
With slight abuse of notation, we denote bidder $i$'s strategy as $\bidi: \typespacei \to \mathbb R_+$.
To distinguish a bidding strategy (a mapping) and a bid (a number), we write the former as $\bidi(\cdot)$ and the latter as $\bidi$.
We write a vector of bids $(\bidi[1](\vali[1]), \cdots, \bidi[n](\vali[n]))$ as $\bids(\vals)$, and denote by $\bidsmi(\valsmi)$ the bids made by bidders except bidder~$i$. 

Bidder~$i$ with value~$\vali$ and bidding~$\bidi$, while the other bidders use bidding strategies $\bidsmi(\cdot)$, 
has \emph{interim utility} 
\begin{align}
\utili(\vali, \bidi, \bidsmi(\cdot)) \coloneqq{}& \Ex[\valsmi\sim \distsmi]{\expostUi(\vali, \bidi, \bidsmi(\valsmi))}     
 \label{eq:interim-util} \\
={}& \Ex[\valsmi \sim \distsmi]{\vali \alloci(\bidi, \bidsmi(\valsmi)) - \payi(\bidi, \bidsmi(\valsmi))}. \nonumber
\end{align}

\begin{defn}
\label{def:bne}
For $\eps \geq 0$, a profile of bidding strategies $\bids(\cdot) = (\bidi[1](\cdot), \cdots, \bidi[n](\cdot))$ in an auction is an \emph{$\eps$-Bayes Nash equilibrium} ($\eps$-BNE) for value distribution $\dists = \prod_{i=1}^n \disti$ if for each bidder~$i\in[n]$ and each $\vali \in \typespacei$, for any $\bidi'\in\mathbb R_+$,
\[
\utili(\vali, \bidi(\vali), \bidsmi(\cdot)) \geq 
\utili(\vali, \bidi', \bidsmi(\cdot)) - \eps.
\]
If $\eps = 0$, $\bids(\cdot)$ is a \emph{Bayes Nash equilibrium}.
\end{defn}

A bidding strategy $\bidi(\cdot)$ is said to be \emph{monotone} if $\val \geq \val' \Rightarrow \bidi(\val) \geq \bidi(\val')$.

\begin{restatable}{prop}{propmonotone}
\label{prop:monotone}
In a first price auction or an all pay auction, for any bidding strategy $\bidi(\cdot)$ of bidder~$i$, for any value distributions $\distsmi$, and any bidding strategies $\bidsmi(\cdot)$, there is a monotone bidding strategy $\bidi'(\cdot)$ such that $\forall \vali \in \typespacei$, $\utili(\vali, \bidi'(\vali), \bidsmi(\cdot)) \geq \utili(\vali, \bidi(\vali), \bidsmi(\cdot))$.
\end{restatable}
\noindent The proof of Proposition~\ref{prop:monotone} is in Appendix~\ref{sec:proof-prop:monotone}. 


Throughout the paper we assume that there is an upper bound $H$ on $\typespacei$. Therefore, no bidder is willing to bid $\bidi >H$ and we only consider bidding strategies $\bidi(\cdot)$ that map $\typespacei$ to $[0, H]$.

\section{Sample Complexity of Utility Estimation}
\label{sec:fpa}


We are given a set of $\nsample$ samples,  $\samples = (\samples^1, \ldots, \samples^{\nsample})$ from the value distribution $\dists = \prod_i \disti$.  
 Note that each sample $\samples^j$ is a vector of values $(\samplei[1]^j, \ldots, \samplei[n]^j)$ drawn from~$\dists$.  

Given value samples~$\samples$, 
a \emph{utility learning algorithm}, denoted by $\alg$,
for each bidder~$i$ with value~$\vali$ and bidding strategies $\bids(\cdot) = (\bidi[1](\cdot), \cdots, \bidi[n](\cdot))$, 
outputs $\alg(\samples, i, \vali, \bids(\cdot))$, which estimates bidder~$i$'s interim utility when her value is~$\vali$ and the bidders use bidding strategies $\bids(\cdot)$. 



\begin{defn}
\label{def:util-learn-ensemble}
Let $\strategyset$ be a set of bidding strategy profiles. 
For $\eps>0, \delta \in (0, 1)$, a utility learning algorithm~$\alg$ \emph{$(\eps, \delta)$-learns with $\nsample$~samples the utilities over $\strategyset$} if,
for any product value distribution~$\dists$,
with probability at least $1 - \delta$, 
for any bidding strategy profile $\bids(\cdot) \in \strategyset$, 
for each bidder~$i\in[n]$ and any $\vali \in \typespacei$, \[
\left| \alg(\samples, i, \vali, \bids(\cdot)) - \utili(\vali, \bidi(\vali), \bidsmi(\cdot)) \right| < \eps,
\]
where the randomness is over the random draw of samples and the randomness of $\alg$ if it is randomized. 
\end{defn}

Throughout the paper we will take $\strategyset$ to be the set of profiles of monotone bidding strategies  
(see Proposition~\ref{prop:monotone}).


\subsection{Upper Bound on Sample Complexity}
\label{sec:fpa-upper-bound}
In this section, we show that $\tilde O(n / \eps^2)$ value samples suffice for learning the interim utilities for all monotone bidding strategies.  
The learning algorithm is the empirical distribution estimator, which outputs expected utilities on the uniform distribution over the samples. 


\begin{defn}
\label{def:emp}
The \emph{empirical distribution estimator}, denoted by $\emp$, 
estimates interim utilities on the uniform distribution over the samples.
Formally, on samples $\samples = (\samples^1, \ldots, \samples^{\nsample})$, for bidder~$i$ with value~$\vali$, for bidding strategies $\bids(\cdot)$,
\begin{equation*}
\emp(\samples, i, \vali, \bids(\cdot)) \coloneqq \frac{1}{\nsample} \sum_{j=1}^\nsample  \expostUi(\vali, \bidi(\vali), \bidsmi(\samplesmi^j)).
\end{equation*}
\end{defn}


One more technicality concerns what happens when more than one bidder makes the highest bid.  
Even though in any equilibrium of a first price auction, this must happen with probability~$0$, for the utility learning problem to be well defined we need to specify tie-breaking rules.
As we shall see, the tie-breaking rule does not affect the sample complexity.
We consider here two tie-breaking rules: by \emph{random-allocation} rule, the item is assigned to one of the highest bidders uniformly at random; 
by \emph{no-allocation} rule, no one wins the item whenever there is a tie for the highest bid.\footnote{
In fact, our results hold for all tie-breaking rules by which any bidder~$i$'s allocation is deterministically determined by the comparisons of her bid $\bidi$ with the other bidders' bids, and not affected by any other information (such as the specific value of~$\bidi$).  
More formally, the allocation for bidder~$i$ is either $0$ or $1$ and is determined by a vector in $\{<, =, >\}^{n-1}$, which records the comparison between $\bidi$ and every other bidder's bid.
Rules satisfying this condition include no-allocation, breaking ties lexicographically, etc.
} 

\begin{thm}\label{thm:util-learn-upper-bound}
Suppose $\typespacei\subseteq[0, H]$ for each $i\in[n]$, and the tie-breaking rule is random-allocation or no-allocation.
For any $\eps>0, \delta \in (0, 1)$, there is
\begin{equation}
M = O \left(\frac{H^2}{\eps^2} \left[n\log n\log \left(\frac{H}{\eps} \right) + \log\left(\frac{n}{\delta}\right)\right]\right),
\label{eq:util-learn-upper-bound}
\end{equation}
such that for any $\nsample \geq M$, 
the empirical distribution estimator $\emp$ $(\eps, \delta)$-learns with $\nsample$ samples
the utilities over the set of all monotone bidding strategies.
\end{thm}





\subsubsection{Pseudo-dimension and the Proof of Theorem \ref{thm:util-learn-upper-bound}}
\label{sec:pseudodim}

Pseudo-dimension is a well known tool for upper bounding sample complexity \citep[see, e.g.][]{anthony2009neural}, and has been applied to learning in mechanism design \citep{MR15, MR16, BSV18, BSV19}.

\begin{defn}
	\label{def:pseudo-dimension}
	Given a class $\funcClass$ of real-valued functions on input space $\PinputSpace$, a set of input $\Pinputi[1], \ldots, \Pinputi[m]$ is said to be \emph{pseudo-shattered} if there exist \emph{witnesses} $\Pwitnessi[1], \ldots, \Pwitnessi[m] \in \mathbb R$ such that for any label vector $\Plabels\in\{1, -1\}^m$, there exists $\func_{\Plabels}\in \funcClass$ such that $\sgn(\func_{\Plabels}(\Pinputi) - \Pwitnessi)  = \Plabeli$ for each $i=1, \ldots, m$, where $\sgn(y)=1$ if $y>0$ and $-1$ if $y<0$. The \emph{pseudo-dimension} of $\funcClass$, $\Pdim(\funcClass)$, is the size of the largest set of inputs that can be pseudo-shattered by $\funcClass$. 
\end{defn}

\begin{defn}
	\label{def:uniform-convergence}
	For $\eps>0, \delta \in (0, 1)$, a class of functions $\funcClass: \PinputSpace \to \mathbb R$ is \emph{$(\eps, \delta)$-uniformly convergent with sample complexity $M$} if
		for any $\nsample \geq M$, 
	for any distribution $\dist$ on~$\PinputSpace$, 
	if $\sample^1, \ldots, \sample^{\nsample}$ are i.i.d.\@ samples from~$\dist$, 
	with probability at least $1 - \delta$, for every $\func \in \funcClass$,
	$
		\left| \Ex[\Pinput \sim \dist]{\func(\Pinput)} - \frac 1 {\nsample} \sum_{j = 1}^{\nsample} \func(\sample^j) \right| < \eps.
	$
\end{defn}

\begin{thm}[See \citealp{anthony2009neural}]
	\label{thm:pseudo-dimension}
	Let $\funcClass$ be a class of functions with range $[0, H]$ and pseudo-dimension $d=\Pdim(\funcClass)$, 
	for any $\eps>0$, $\delta\in(0, 1)$, 
	$\funcClass$ is $(\eps, \delta)$-uniformly convergent with sample complexity $O\left( (\frac{H}{\eps})^2 [d\log(\frac{H}{\eps}) + \log(\frac{1}{\delta})] \right)$.
\end{thm}

We show Theorem~\ref{thm:util-learn-upper-bound} by treating the utilities on monotone bidding strategies as a class of functions, whose uniform convergence implies that $\emp$ learns the interim utilities.


For each bidder $i$, let $\func^{\vali, \bids(\cdot)}$ be the function that maps the opponents' values to bidder~$i$'s ex post utility, that is, 
\[\func^{\vali, \bids(\cdot)}(\valsmi) = \expostUi(\vali, \bidi(\vali), \bidsmi(\valsmi)).\]
Let $\funcClass_i$ be the set of all such functions corresponding to the set of monotone strategies, 
\[\funcClass_i = \left\{\func^{\vali, \bids(\cdot) }(\cdot) \given \vali \in \typespacei,~~ \bids(\cdot) \text{ is monotone} \right\}. \]

By \eqref{eq:interim-util}, the expectation of $\func^{\vali, \bids(\cdot)}(\cdot)$ over~$\distsmi$ is the interim utility of bidder $i$: 
\[\Ex[\valsmi\sim \distsmi]{\func^{\vali, \bids(\cdot)}(\valsmi)} = \utili(\vali, \bidi(\vali), \bidsmi(\cdot)). \]
By Definition~\ref{def:emp}, on samples $\samples = (\samples^1, \ldots, \samples^{\nsample})$, 
\[\emp(\samples, i, \vali, \bids(\cdot)) = \frac 1 {\nsample} \sum_{j = 1}^{\nsample} \func^{\vali, \bids(\cdot)}(\samples^j_{-i}).
\]

Thus, 
\begin{align}
  \left| \emp(\samples, i, \vali, \bids(\cdot)) -  \utili(\vali, \bidi(\vali), \bidsmi(\cdot)) \right| 
 =
        \left| \Ex[\valsmi]{\func^{\vali, \bids(\cdot)}(\valsmi)} - \frac 1 {\nsample} \sum_{j = 1}^{\nsample} \func^{\vali, \bids(\cdot)}(\samples^j_{-i})\right|.
	\label{eq:emp-func}
\end{align}

The right hand side of~\eqref{eq:emp-func} is the difference between the expectation of $\func^{\vali, \bids(\cdot)}$ on the distribution $\distsmi$ and that on the empirical distribution with samples drawn from $\distsmi$.
Now by Theorem~\ref{thm:pseudo-dimension},
to bound the number of samples needed by $\emp$ to $(\eps, \delta)$-learn the utilities over monotone strategies, 
it suffices to bound the pseudo-dimension of~$\funcClass_i$.
With the following key lemma, the proof is completed by observing that the range of each $\func^{\vali, \bids(\cdot)}$ is within $[-H, H]$ and by taking a union bound over $i \in [n]$.



\begin{restatable}{lemma}{pdimutil}
\label{lem:pseudo-dimension-utility-class}
If tie breaking is random-allocation or no-allocation, then $\Pdim(\funcClass_i) = O(n \log n)$.
\end{restatable}

The proof of Lemma~\ref{lem:pseudo-dimension-utility-class} follows a powerful framework introduced by \citet{MR16} and \citet{BSV18} for bounding the pseudo-dimension of a class $\funcClass$ of functions: given samples that are to be pseudo-shattered and for any (fixed) witnesses, one classifies the functions in~$\funcClass$ into categories, so that functions in the same category must output the same label on all the samples; by counting and bounding the number of such categories, one can bound the number of shattered samples.
Our proof follows this strategy.  To bound the number of categories, we make use of monotonicity of bidding functions, which is specific to our problem.

We give a proof below for the simplest case with two bidders and no-allocation tie-breaking rule, and relegate the full proof to Appendix~\ref{sec:proof-lem:pseudo-dimension-utility-class}. 

\begin{proof}[Proof of Lemma~\ref{lem:pseudo-dimension-utility-class} for a special case.] 
Consider $n=2$ and no-allocation tie-breaking rule. Fix an arbitrary set of $\nsample$ samples $\samplesmi^1, \ldots, \samplesmi^{\nsample}$.  
Consider any set of potential witnesses $(\Pwitnessi[1], \Pwitnessi[2], \ldots, \Pwitnessi[\nsample])$. 
Each hypothesis in $\funcClass_i$ then gives every sample~$\samplesmi^j$ a label according to the witness~$\Pwitnessi[j]$, giving rise to a label vector in $\{-1, +1\}^\nsample$.
We show that $\funcClass_i$ can be divided into $\nsample+1$ sub-classes $\funcClass_i^0, \ldots, \funcClass_i^{\nsample}$, such that each sub-class $\funcClass_i^k$ generates at most $\nsample+1$ different label vectors. 
Thus $\funcClass_i$ generates at most $(\nsample+1)^2$ label vectors in total. 
To pseudo-shatter $\samplesmi^1, \ldots, \samplesmi^\nsample$, we need $2^\nsample$ different label vectors; therefore $(m+1)^2\ge 2^\nsample$, which implies $\nsample = O(1)$. 

We now show how $\funcClass_i$ is thus divided.
Note that, for $n = 2$, each $\samplesmi^k$ is just a real number and we can sort them; for ease of notation let 
$\Pinput^k$ denote $\samplei[-i]^{k}$ for $k=1, \ldots, \nsample$ 
and suppose $\Pinput^1\le \Pinput^2\le \cdots \le \Pinput^\nsample$. 
We put hypothesis $\func^{\vali, \bids(\cdot)}$ into the $k^{\text{-th}}$ sub-class, $\funcClass_i^k$, if
\[ \bidi[-i](\Pinput^k) < \bidi(\vali)  \ \text{ and }\ \bidi[-i](\Pinput^{k+1}) \ge \bidi(\vali).\]
This is well defined because, by assumption, $\bidi[-i](\Pinput)$ is monotone non-decreasing in $\Pinput$.

We now show that each sub-class $\funcClass_i^k$ gives rise to at most $\nsample+1$ label vectors.
For any $\func^{\vali, \bids(\cdot)}\in \funcClass_i^k$, we have $\func^{\vali, \bids(\cdot)}(\Pinput^{j}) =  \vali - \bidi(\vali)$ for any $j \le k$ (because bidder~$i$'s bid $\bidi(\vali)$ is higher than the opponent's),
and $\func^{\vali, \bids(\cdot)}(\Pinput^{j}) =  0$ for any $j > k$.
On the first $k$ samples $\Pinput^{1}, \ldots, \Pinput^{k}$, 
any fixed hypothesis $\func^{\vali, \bids(\cdot)}(\cdot) \in 
\funcClass_i^k$ 
outputs a constant $\vali - \bidi(\vali)$;
as one varies this constant and compares it with the $k$ witnesses $\Pwitnessi[1], \ldots, \Pwitnessi[k]$, there are only $k+1$ possible results from the comparisons.
On the remaining $\nsample - k$ samples, only one pattern is possible, since all hypotheses in $\funcClass_i^k$ output $0$ on these samples.
Therefore, at most $k+1 \le \nsample+1$ label vectors can be generated by $\funcClass_i^k$. 
\end{proof}





\subsubsection{Learning on Empirical Product Distributions and Equilibrium Preservation}
\label{sec:empp}

The empirical distribution estimator approximates interim utilities with high probability, but this does not immediately imply that one may take the first price auction on the empirical distribution as a close approximation to the auction on the original distribution. 
This is because the empirical distribution over samples is \emph{correlated} --- the values $\samplei[1]^j, \ldots, \samplei[n]^j$ are drawn as a vector, instead of independently. 
Standard notions, such as Bayes Nash equilibria, defined on product distributions become intricate on correlated distributions, and there is no reason to expect the latter to correspond to the equilibria in the original auction.
Therefore, it is desirable that utilities can also be learned on a \emph{product} distribution arising from the samples, where each bidder's value is independently drawn, uniformly from the $\nsample$ samples of her value.  
We show that this can indeed be done, without substantial increase in the number of samples.
The key technical step, Lemma~\ref{lem:relation-uniform-convergence},
is a reduction from learning on empirical distribution to learning on empirical product distribution.  We believe this lemma is of independent interest.
In fact, in Section~\ref{sec:search} we invoke Lemma~\ref{lem:relation-uniform-convergence} in a different context, that of learning in Pandora's Box problem; the reduction is crucial there for obtaining a polynomial-time learning algorithm.

\begin{defn}
\label{def:empp}
Given samples $\samples = (\samples^1, \ldots, \samples^{\nsample})$, 
$\empDisti$ is defined to be the uniform distribution over $\{\samplei^1, \ldots, \samplei^\nsample\}$. The \emph{empirical product distribution} $\empDists$ is the product distribution
    $\empDists\coloneqq\prod_{i=1}^n \empDisti$.
\end{defn}

	\begin{defn}
		\label{def:uniform-convergence-product}
		For $\eps>0, \delta \in (0, 1)$, a class of functions $\funcClass: \prod_{i=1}^n \typespacei \to \mathbb R$ is \emph{$(\eps, \delta)$-uniformly convergent on product distribution with sample complexity $M$} if
		for any $\nsample \geq M$, 
	for any product distribution $\dists$ on~$\prod_{i=1}^n \typespacei$, 
	if $\samples^1, \ldots, \samples^{\nsample}$ are i.i.d.\@ samples from~$\dists$, 
	with probability at least $1 - \delta$, for every $\func \in \funcClass$,
	\[
		\left| \Ex[\types \sim \dists]{\func(\types)} - \Ex[\types \sim \empDists]{\func(\types)} \right| < \eps,
	\]
	where $\empDists$ is the empirical product distribution. 
	\end{defn}

\begin{restatable}{lemma}{uniformprod}
	\label{lem:relation-uniform-convergence}
Let $\funcClass$ be a class of functions from a product space $\typespaces$ to $[0, H]$. 
If $\funcClass$ is $(\eps, \delta)$-uniformly convergent with sample complexity $\nsample=\nsample(\eps, \delta)$, then $\funcClass$ is $\left(2\eps, \frac{H\delta}{\eps}\right)$-uniformly convergent on product distribution with sample complexity $\nsample$. 
\end{restatable}
\noindent Lemma~\ref{lem:relation-uniform-convergence} is closely related to a concentration inequality by \citet{DHP16}.
\citeauthor{DHP16} show that for any single function $h:\typespaces\to[0, H]$, the expectation of $h$ on the empirical product distribution is close to its expectation on any product distribution with high probability.
Our lemma generalizes this to show a simultaneous concentration for a family of functions,   
and seems more handy for applications such as ours.

Combining Theorem~\ref{thm:util-learn-upper-bound} with Lemma~\ref{lem:relation-uniform-convergence}, we derive our learning results on the empirical product distribution.

\begin{defn}
The \emph{empirical product distribution estimator} $\empp$ estimates interim utilities of a bidding strategy on the empirical product distribution~$\empDists$.  Formally, 
for bidder~$i$ with value~$\vali$, for bidding strategy profile $\bids(\cdot)$,
\begin{equation}
\empp(\samples, i, \vali, \bids(\cdot)) \coloneqq 
\Ex[\valsmi\sim \empDistsmi] { \expostUi(\vali, \bidi(\vali), \bidsmi(\valsmi)) }. \label{eq:def_empp}
\end{equation}
\end{defn}

\begin{thm}\label{thm:util-learn-upper-bound-product}
Suppose $\typespacei\subseteq[0, H]$ for each $i\in[n]$, and the tie-breaking rule is random-allocation or no-allocation.
For any $\eps>0, \delta \in (0, 1)$, there is
\begin{equation}
M = O \left(\frac{H^2}{\eps^2} \left[n\log n\log \left(\frac{H}{\eps} \right) + \log\left(\frac{n}{\delta}\right)\right]\right),
\label{eq:util-learn-upper-bound-product}
\end{equation}
such that for any $\nsample \geq M$, 
the empirical distribution estimator $\empp$ $(\eps, \delta)$-learns with $\nsample$ samples
the utilities over the set of all monotone bidding strategies.
\end{thm}


By Theorem~\ref{thm:util-learn-upper-bound-product}, utilities in the FPA on the empirical product distribution approximate those in the FPA on the original distribution, therefore the two auctions share the same set of approximate equilibria:

\begin{corollary}\label{cor:find-BNE}
Suppose $\typespacei\subseteq[0, H]$ for each $i\in[n]$ and the tie-breaking rule is random-allocation or no-allocation. 
For any $\eps, \eps'>0, \delta \in (0, 1)$, for $m$ satisfying~\eqref{eq:util-learn-upper-bound-product}, 
 with probability at least $1-\delta$ over random draws of $\samples$, for any monotone bidding strategy profile $\bids(\cdot)$, if $\bids(\cdot)$ is an $\eps'$-BNE in the first price auction on value distribution $\empDists=\prod_i\empDisti$, then $\bids(\cdot)$ is an $(\eps'+2\eps)$-BNE in the first price auction on value distribution $\dists = \prod_i \disti$. 
 Conversely, if $\bids(\cdot)$ is an $\eps'$-BNE in the first price auction on value distribution~$\dists$, then $\bids(\cdot)$ is an $(\eps'+2\eps)$-BNE in the first price auction on value distribution~$\empDists$. 
\end{corollary}



Corollary~\ref{cor:find-BNE} has an interesting consequence: if there exists an algorithm that can compute Bayes Nash equilibria for the first price auction on \emph{discrete} value distributions, then there also exists an algorithm that can compute approximate Bayes Nash equilibria for \emph{any} distributions, by simply sampling from the distribution and running the former algorithm on the empirical product distribution (which is discrete).\footnote{In particular, \citet{SWZ20} give an ``algorithm'' that outputs a strategy profile that approaches a Bayes Nash equilibrium in the first price auction on any discrete distributions if the algorithm is allowed to run for infinite time.  It is unclear whether their algorithm can output an $\eps$-Bayes Nash equilibrium if run for finite time.}
\begin{corollary}
\label{cor:polytime-equilibria}
If there exists an algorithm that computes monotone $\eps'$-BNE for the first price auction on any discrete product value distributions $\bm{D} = \prod_{i=1}^n D_i$, then there exists a sample-access algorithm that computes $(\eps'+\eps)$-BNE for the first price auction on any product distributions $\bm F = \prod_{i=1}^n F_i$ with high probability.
If the running time of the former algorithm is polynomial in $\frac{1}{\eps'}$ and the support size of each component distribution $D_i$, then the running time of the latter algorithm is $\mathrm{poly}(\frac{1}{\eps}, \frac{1}{\eps'})$, which does not depend on the support size of $F_i$ (and $F_i$ can be continuous). 
\end{corollary}

Results very similar to Theorem~\ref{thm:util-learn-upper-bound}, Theorem~\ref{thm:util-learn-upper-bound-product}, Corollary~\ref{cor:find-BNE}, and Corollary~\ref{cor:polytime-equilibria}, apply to the all pay auction, with the same bounds on the number of samples. 
The proofs are almost identical and so are omitted.

\subsection{Lower Bound of Sample Complexity}
\label{sec:lower-bound}

We give an information theoretic lower bound on the number of samples needed for any algorithm to learn monotone utilities in a first price auction. 
The lower bound matches our upper bound up to polylog factors.

\begin{restatable}{thm}{lowerbound}
\label{thm:lower-bound-learning-util}
For any $\eps < \frac 1 {4000}, \delta < \frac 1 {20}$, 
there is a family of product distributions for which no algorithm $(\eps, \delta)$-learns, with $\nsample$ samples, utilities over the set of all monotone bidding strategies, for any $m \leq \frac{1}{4\times10^8}\cdot \frac{n}{\eps^2}$. 
\end{restatable}

The proof of Theorem~\ref{thm:lower-bound-learning-util} is in Appendix~\ref{sec:proof-thm:lower-bound-learning-utility}.  
As a sketch, the product distributions we construct encode length~$n-1$ binary strings by having slightly unfair Bernoulli distribution for each bidder, the bias shrinking as $n$ grows large.
We then show that, if with few samples a learning algorithm learns utilities for all monotone bidding strategies, then there must exist two product distributions from the family which differ at only one coordinate, and yet they can be told apart by the learning algorithm. 
This must violate the well-known information theoretic lower bound for distinguishing two distributions \cite[e.g.][]{MansourNotes}.

\section{Auctions with Costly Search}
\label{sec:search}
We extend our sample complexity results to auctions in which bidders need to incur a cost to know precisely their values, a model proposed and studied by \citet{KWW16}.

In this model, each bidder~$i$ knows the distribution~$\disti$ from which her value is drawn, but gets to know her value~$\vali$ only after incurring a cost~$\costi$.
This models well, for example, a real estate market, where $\costi$ is an inspection cost.  
\citet{KWW16} showed that, due to the search costs, the English auction can have low efficiency, whereas the Dutch auction, with its descending price, can coordinate the bidders' searching in an almost efficient way.  
Intuitively, a bidder does not inspect her value until the price drops to a certain level, and then, after inspection at this threshold, either claims the item at the threshold price, or waits till later.  
In fact, absent incentive issues, this is the procedure a central authority would follow to maximize the welfare; the elegant algorithm is known as the Pandora's Box algorithm \citep{Weitzman79}.  
With incentives, bidders shade their bids just as in a first price auction, and there is efficiency loss. 
This was made precise by \citeauthor{KWW16}, who showed a correspondence between the equilibria in a Dutch auction with search costs and the equilibria in a first price auction without search costs but with transformed value distributions.
The near efficiency of the Dutch auction therefore follows from Price of Anarchy results on the first price auction \citep{ST13, HTW18, jin_first_2022}.

In this section, we first review in Section~\ref{sec:pandora} Pandora's Box algorithm, necessary for understanding the correspondence observed by \citet{KWW16}.  
En route, we show that $\tilde O(n / \eps^2)$ samples from the value distributions suffice for the algorithm to be $\eps$-close to optimal when the distributions are unknown.  
Our bound slightly improves a recent result by \citet{guo_generalizing_2021}.

We then review, in Section~\ref{sec:fpa-pandora}, the correspondence between the Dutch auction with search costs and the FPA without search costs.  
The correspondence between auctions involves  a mapping between the strategies in the two auctions and a transformation of the value distributions.  The mapping and transformation depend on the  value distributions and the search costs. 
We show that, when the value distributions are unknown, with $\tilde O(1 / \eps^2)$ value samples, an ``empirical correspondence'' can be established such that all monotone bidding strategies in the Dutch auction have approximately the same utilities as their image bidding strategies in the FPA.
Combining with our learning results on the FPA, we show that, with $\tilde O(n / \eps^2)$ samples, any equilibrium of the FPA without search costs on a transformed empirical distribution can be mapped to an approximate equilibrium of the Dutch auction on the true distribution.
Moreover, due to the bounded Price of Anarchy of the FPA \citep[e.g.][]{ST13}, the Dutch auction with this approximate equilibrium is approximately optimal in terms of welfare.



\subsection{Pandora’s Box Problem and Its Sample Complexity}
\label{sec:pandora}
Absent search costs, the welfare (a.k.a.\@ the efficiency) of a single item auction is the value of the bidder who is allocated the item.  
The maximum expected welfare is therefore simply the expectation of the largest value among the bidders.  
Auctions that sell to the highest bidder and charges the winner a price equal to the second highest bid gives bidders correct incentives to bid their true values and maximizes the welfare.
The sealed-bid second price auction and the ascending price auction (the English auction) both achieve this, whereas the first price auction and the equivalent descending price auction (the Dutch auction) obtain only approximately optimal welfare at equilibrium \citep{ST13, HTW18, jin_first_2022}.
With search costs, the welfare of an auction is the value of the bidder winning the item minus all the search costs paid.
Even without incentive considerations, the problem is nontrivial algorithmically.

\paragraph{The Pandora's Box Problem.}
The following Pandora's Box problem, named by \citet{Weitzman79}, abstracts the welfare maximization problem in the presence of search costs.
We are given $n$ boxes, each box~$i$ containing a value~$\vali$ drawn independently from a known distribution~$\disti$;  
to open box~$i$ and see~$\vali$, we must pay a cost of $\costi$;
at any point, we can take any box that has been opened and quit, or open a closed box at a cost, or quit without taking anything.
Our payoff is the value in the box taken (if any) minus the costs we paid along the way.
Given $\disti[1], \ldots, \disti[n]$ and $\costi[1], \ldots, \costi[n]$, our goal is to compute a procedure that maximizes the expected payoff.

\citet{Weitzman79} used this setting to model a consumer searching for an item to purchase; he gave an optimal algorithm, which is in turn a special case of Gittins Index algorithm from Bayesian bandits \citep{Gittins79}.  

We describe his algorithm below.  
To facilitate discussion of learning, we treat search costs as given, and algorithms as mappings from (unseen) values $\vali[1], \ldots, \vali[n]$ to a payoff.  
Certainly, only mappings that correspond to valid search procedures are meaningful; in particular, the procedure's decision (e.g., to open which box) cannot depend on values that have not been revealed.
It is the associated search procedure that we are interested in.

\begin{defn}[Index Based Algorithms/Mappings]
\label{def:index-based}
Given search costs $(\costi[1], \ldots, \costi[n])$, a mapping~$\alg$ from $(\vali[1], \ldots, \vali[n]) \in [0, \highval]^n$ to~$\mathbb R$ 
is \emph{index based} if there exist \emph{indices} $\reservei[1], \ldots, \reservei[n] \in \mathbb R$ such that on any vector of values $(\vali[1], \ldots, \vali[n])$, the output of~$\alg$ is given by the following procedure:
\begin{itemize}
\item Open boxes in descending order of $\reservei$.  Without loss of generality assume $\reservei[1]\ge \cdots \ge \reservei[n]$.
\item If $\reservei[1]<0$, do not open any box and output~$0$.
\item Otherwise, after opening box $i$, if $\max\{\vali[1], \ldots, \vali[i]\} \ge \reservei[i+1]$ or $i=n$, stop immediately and output $\max\{\vali[1], \ldots, \vali[i]\}-(c_1+\cdots+c_i)$. 
\end{itemize}
\end{defn}

\begin{thm}[\citealp{Weitzman79}] \label{thm:optimal-pandora}
The optimal algorithm corresponds to an index-based mapping; 
the index $\reservei$ for box~$i$ is the unique solution to $\Ex[\val \sim \disti]{\max(\val - \reservei, 0)} = \costi$.
\end{thm}

\paragraph{Learning.}
 We now answer the following learning question: if the distributions $\disti[1], \cdots, \disti[n]$ are unknown, how many samples from them suffice for us to devise an algorithm that is close to optimal on the original distribution?
Recently, \citet{guo_generalizing_2021} gave a polynomial bound for the problem; we give an alternative analysis using pseudo-dimension, which leads to a slightly improved bound.  
We make use of a technical lemmas of theirs (Lemma~\ref{lem:truncate}).
For our learning algorithm to be run in polynomial time, we invoke Lemma~\ref{lem:relation-uniform-convergence} to perform learning on the empirical product distribution.

Given our view of the algorithms as mappings from value vectors to the payoff, the expected payoff of an algorithm is then the expectation of its output on the value distributions.  
Given Theorem~\ref{thm:optimal-pandora}, it suffices to learn the expected payoff of all index-based algorithms.  
The problem then boils down to bounding the pseudo-dimension of the class of index-based mappings.
Modulo a technical issue which calls for truncating the index-based algorithms, that is an outline of the proof of the following sample complexity theorem.

\begin{thm}
	\label{thm:pandora}
	Given search costs $\costi[1], \ldots, \costi[n]$, for any $\eps > 0, \delta \in (0, 1)$, there is 
	\[M = O\left(\frac{H^2}{\eps^2}\log^2 (\frac H \eps) \left[n\log n\log (\frac H \eps) + \log (\frac {1}{\delta})\right]\right),\]
	such that given $\nsample \ge M$ samples, 
	a search procedure computed on these samples has expected payoff within additive $\eps$ to the optimal algorithm with probability at least $1 - \delta$.
	Moreover, the procedure can be computed in polynomial time.
\end{thm}
Compared with \citet{guo_generalizing_2021}'s bound $O(\frac n {\eps^2} \log^2 (\frac 1 \eps) \log (\frac n \eps) \log (\frac n {\eps \delta}))$ (where $\highval$ is normalized to~$1$), our bound is better: 
theirs has a $\frac{n}{\eps^2}\log^2(\frac{1}{\eps})(\log^2 n + \log\frac{1}{\eps}\log\frac{1}{\eps\delta})$ term while we do not.

We devote the rest of this subsection to the proof of Theorem~\ref{thm:pandora}.  
Let $\funcClass_P$ be the class of all index-based mappings.  
The technical centerpiece is a bound on the pseudo-dimension of~$\funcClass_P$.



\begin{restatable}{lemma}{pdimpandora}
\label{lem:pandora-pdim}
$\PDim(\funcClass_P) = O(n \log n)$.
\end{restatable}
\begin{proof}
Given any profile of values $(\vali[1], \ldots, \vali[n]) \in \mathbb [0, H]^n$, the output of any index-based mapping with indices $(\reservei)_i$ is fully determined by the following $O(n^2)$ linear inequalities: for any $i, j \in [n]$, whether $\reservei \geq \reservei[j]$ or $\reservei < \reservei[j]$; for any $i, j \in [n]$, whether $\reservei \geq \vali[j]$ or $\reservei < \vali[j]$.  
That is, the space of indices is partitioned by the hyperplanes given by these $O(n^2)$ inequalities, and within each region the corresponding index-based mapping remains a constant for this profile of values.  
Consider any $m$ value profiles that are pseudo-shattered by~$\hypoclass_P$.  
Each of these $m$ value profiles imposes $O(n^2)$ linear inequalities on the space of indices, and we will have altogether $O(mn^2)$ inequalities.  
A crucial observation is that, for any positive integer~$t$, the space $\mathbb R^n$ can be partitioned by $t$ hyperplanes into at most $O(t^n)$ regions.
Therefore the space of indices, which is $\mathbb R^n$, can be divided into at most $(C mn^2)^n$ regions, for some constant $C > 0$.  
Any index-based algorithm within such a region gives the same outputs on all these $m$ value profiles, and therefore cannot give different signs for any profile no matter what the corresponding witness is.  
To shatter $m$ profiles we need at least $2^m$ regions.  
Therefore $2^m \leq (Cmn^2)^n$, which gives $m \leq C' n \log n$ for some $C' >0$.
\end{proof}

Note that, if the values are between $0$ and~$\highval$, without loss of generality we may assume $\costi \leq \highval$ for each~$i$. (Otherwise the box should be discarded by any reasonable algorithm.)
With this, directly combining Lemma~\ref{lem:pandora-pdim} and Theorem~\ref{thm:pseudo-dimension} would still yield a bound having a cubic dependence on~$n$, because the output of an index-based mapping may span the range $[-n\highval, \highval]$.  
A similar problem also arose in the approach of \citet{guo_generalizing_2021}, who remedied this by observing that the performance of the optimal index-based algorithm is not affected much if it is truncated: to \emph{truncate} an algorithm for the Pandora's Box problem, the algorithm is terminated immediately when its cumulated cost exceeds $\Omega(H\log \frac H \eps)$.

\begin{lemma}[\citealp{guo_generalizing_2021}]
	\label{lem:truncate}
On an instance of the Pandora's Box problem, the expected payoff of the optimal index-based algorithm exceeds that of its truncated version by no more than~$\eps$.
\end{lemma}

The proof of Lemma~\ref{lem:pandora-pdim} is easily modified to give the same bound on the pseudo-dimension of mappings corresponding to truncated index-based algorithms.  
With this, we can now combine Theorem~\ref{thm:pseudo-dimension} and Lemma~\ref{lem:relation-uniform-convergence} to obtain a sample complexity upper bound.

We remark that in Theorem~\ref{thm:pandora} we show the sample complexity for uniform convergence \emph{on product distribution}, because this yields a fast algorithm given samples: simply running the optimal truncated index-based algorithm on the empirical product distribution is guaranteed to be approximately optimal on~$\dists$ with high probability.  
On the other hand, picking out the best index-based algorithm on the empirical distribution, which is a correlated distribution, appears computationally challenging.

\subsection{Descending Auction with Search Costs}
\label{sec:fpa-pandora}
In this section, we briefly review the main results by \citet{KWW16} in Section~\ref{sec:KWW16}, and then in Section~\ref{sec:sample-KWW16} present our learning results in auctions with search costs.
Recall that in this setting,  we consider a single-item auction, where each bidder~$i$ has a value~$\vali \in [0, H]$ drawn independently from distribution~$\disti$, but 
$\vali$ is not known to anyone at the beginning of the auction. 
In order to observe the value, bidder~$i$ needs to pay a known search cost $\costi\in[0, \highval]$.

\subsubsection{Transformation with Distributional Knowledge}
\label{sec:KWW16}

\paragraph{Descending auction with search costs.}  
In a \emph{descending auction} (or Dutch auction), a publicly visible price descends continuously from~$H$.  
At any point, any bidder may claim the item at the current price.  
With search cost, a bidder's strategy $\dstrati$ consists of two parts:\footnote{Note that there is no private information at the beginning of the auction.} a threshold price $\dtimei$ and a mapping $\dbidi(\cdot)$ from values to bids.  
Concretely, bidder~$i$ decides to inspect when the price descends to~$\dtimei$, at which point she pays the search cost and immediately learns her value~$\vali$.  
After seeing her value, the bidder chooses another a purchase price $\dbidi(\vali) \leq \dtimei$ at which to claim the item.
The latter is equivalent to submitting a bid $\dbidi(\vali)\le\dtimei$. 

We say a strategy $\dstrati=(\dtimei, \dbidi(\cdot))$ is \emph{monotone} if $\dbidi(\cdot)$ is monotone non-decreasing. A strategy is \emph{mixed} if it is a distribution over pure strategies $\dstrati$'s. Mixed strategies allow bidders to randomize over the threshold price $\dtimei$ and the purchase price $\dbidi(\vali)$. Abusing notations, we also use $\dstrati$ to denote a mixed strategy.
We say a \emph{mixed} strategy $\dstrati$ is \emph{monotone} if it is a distribution over monotone pure strategies. 

We use $\DA(\dists, \costs)$ to denote a descending auction on value distributions $\dists$ with search costs $\costs$, and let $\utili^{\DA(\dists, \costs)}(\dstrati, \dstratsmi)$ be the expected utility of bidder $i$ when bidders use strategies $\dstrats=(\dstrati, \dstratsmi)$ and their values are drawn from $\dists$. Note that this utility is ex ante, since the value is unknown until the bidder searches. 
The solution concept we consider is therefore a Nash equilibrium instead of a Bayes Nash equilibrium. 

\begin{defn} 
In $\DA(\dists, \costs)$, a (mixed) strategy profile $\dstrats$ is an \emph{$\eps$-Nash equilibrium} ($\eps$-NE) if for each bidder~$i$ and any strategy $\dstrati'$, 
\[\utili^{\DA(\dists, \costs)}(\dstrati, \dstratsmi) \ge \utili^{\DA(\dists, \costs)}(\dstrati', \dstratsmi) - \eps.\]
If $\eps = 0$, $\dstrats$ is a \emph{Nash equilibrium}.
\end{defn}

We use $\FPA(\dists)$ to denote the first price auction with value distributions~$\dists$.  
Denote by $\utili^{\FPA(\dists)}(\fstrats)$ the (ex ante) expected utility of bidder~$i$ in $\FPA(\dists)$, when the bidders use strategy profile~$\fstrats$.  
We can similarly define the Nash equilibrium for a first price auction. 
\begin{defn} 
In $\FPA(\dists)$, a (mixed) strategy profile $\fstrats$ is an \emph{$\eps$-Nash equilibrium} ($\eps$-NE) if for each bidder~$i$ and any strategy $\fstrati'$,
\[\utili^{\FPA(\dists)}(\fstrati, \fstratsmi) \ge \utili^{\FPA(\dists)}(\fstrati', \fstratsmi) - \eps.\]
If $\eps = 0$, $\fstrats$ is a \emph{Nash equilibrium}.
\end{defn}

Note that Nash equilibrium is an ex ante notion, in contrast to BNE (Definition~\ref{def:bne}), which is an interim notion and requires every type to best respond.
In a first price auction, an $\eps$-BNE  must be an $\eps$-NE, but the reverse is not true.

With no search cost, the Dutch auction is well known to be equivalent to a first price auction. 
Given a Dutch auction with search costs,   \citet{KWW16} constructed a first price auction with transformed value distributions and no search costs, and showed that an NE of this FPA corresponds to an NE of the Dutch auction with search costs.  


\begin{defn}\label{def:fpa-pandora-index}
Given a distribution $\disti$ and a search cost $\costi$, define the \emph{index} $\thresholdi$ of $(\disti, \costi)$ to be the unique solution to $\Ex[\vali\sim\disti]{ \max\{\vali - \thresholdi, 0\} } = \costi$.  If $\costi=0$, let $\thresholdi=H$. 
We always assume $\Ex[\vali\sim\disti]{\vali}\ge \costi$, so that $\thresholdi\in[0, H]$.  (Otherwise the search cost would be so high that the bidder should never search for the value.)
\end{defn}

For a distribution $\dist$ and some $\threshold\in \mathbb R$, we define 
$\dist^{\threshold}$ to be the distribution of $\min\{\val, \threshold\}$ where $\val\sim F$. 
For a product distribution~$\dists$ and a vector $\thresholds$, we use $\dists^{\thresholds}$ to denote the product distribution whose $i^{\text{-th}}$ component is $\disti^{\thresholdi}$.  
A key insight of \citet{KWW16} is a pair of utility-preserving mappings between strategies in DA$(\dists, \costs)$ and FPA$(\dists^{\thresholds})$, where $\thresholds$ is the vector of indices for $(\dists, \costs)$.


\begin{defn}\label{def:strategy-mappings}
For each bidder $i$, given distribution $\disti$ and $\thresholdi \in [0, H]$, define two mappings:\footnote{We describe mappings for pure strategies here.  For mixed strategies, their images are naturally distributions over the images of pure strategies under $\FtD$ and $\DtF$.}
\begin{enumerate} 
\item $\FtD^{\thresholdi}$: a monotone strategy $\fstrati:[0, \thresholdi]\to\mathbb R_+$ in $\FPA(\dists^{\thresholds})$ is mapped by $\FtD^{\thresholdi}$ to 
the strategy in $\DA(\dists, \costs)$ with threshold price $\dtimei=\fstrati(\thresholdi)$ and bidding function $\dbidi(\vali) = \fstrati(\min \{\vali, \thresholdi\})$. 
(By the monotonicity of~$\fstrati$, we have $\dbidi(\vali)\le \dtimei$). 

\item $\DtF^{(\disti, \thresholdi)}$: a strategy $\dstrati = (\dtimei, \dbidi(\cdot))$ in $\DA(\dists, \costs)$ is mapped by $\DtF^{(\disti, \thresholdi)}$ to a strategy $\fstrati=\DtF^{(\disti, \thresholdi)}(\dstrati)$ in $\FPA(\dists^{\thresholds})$, with $\fstrati(\vali) = \dbidi(\vali)$ for $\vali<\thresholdi$ and $\fstrati(\thresholdi)=\dbidi(\vali')$, where $\vali'$ is a random variable drawn from~$\disti$ conditioning on $\vali' \ge \thresholdi$.
\end{enumerate}
\end{defn}

\noindent The superscripts $\thresholdi$ and $(\disti, \thresholdi)$ should make it clear that the mapping~$\FtD^{\thresholdi}$ is determined solely by~$\thresholdi$ while $\DtF^{(\disti, \thresholdi)}$ is related to both the distribution and~$\thresholdi$. 

A strategy $\dstrati = (\dtimei, \dbidi(\cdot))$ in a descending auction is said to \emph{claim above $\thresholdi$}  
if $\dbidi(\vali) =\dtimei$ for all $\vali\ge \thresholdi$, i.e., the bidder claims the item immediately if she finds the value of the item greater than or equal to~$\thresholdi$. 

\begin{claim}[Claim 2 of \citealp{KWW16}]
\label{claim:strategy-equivalence}
Given distribution $\disti$ whose index is $\thresholdi$, 
\begin{enumerate}
    \item If $\dstrati$ claims above $\thresholdi$, then $\dstrati = \FtD^{\thresholdi}(\DtF^{(\disti, \thresholdi)}(\dstrati))$. 
    \item If $\fstrati$ is monotone, then $\fstrati = \DtF^{(\disti, \thresholdi)}(\FtD^{\thresholdi}(\fstrati))$. 
\end{enumerate}
\end{claim}

\begin{thm}[Claim 3 of \citealp{KWW16}]
\label{thm:DA_FPA_transform}
Suppose $\thresholds$ is the vector of indices of $(\dists, \costs)$ (Definition~\ref{def:fpa-pandora-index}). 
\begin{enumerate}
 \item For any monotone mixed strategy profile $\fstrats = (\fstrati, \fstratsmi)$ for $\FPA(\dists^{\thresholds})$,
 for each bidder~$i$, 
\[  \utili^{\FPA(\dists^{\thresholds})}(\fstrats) = \utili^{\DA(\dists, \costs)}(\FtD^{\thresholds}(\fstrats)).\]
\item For any mixed (not necessarily monotone) strategy profile $\dstrats = (\dstrati, \dstratsmi)$ for $\DA(\dists, \costs)$, for each bidder~$i$,
\[ \utili^{\DA(\dists, \costs)}(\dstrats) \le \utili^{\FPA(\dists^{\thresholds})}(\DtF^{(\dists, \thresholds)}(\dstrats)),\]
where ``='' obtains if $\dstrati$ claims above $\thresholdi$.
\end{enumerate}
\end{thm}


\begin{thm}[\citealp{KWW16}]\label{thm:fpa-pandora-NE-BNE}
Given $\DA(\dists, \costs)$ and $\FPA(\dists^{\thresholds})$ where $\thresholds$ is the indices of $(\dists, \costs)$.
If $\fstrats$ is a BNE in $\FPA(\dists^{\thresholds})$, then $\FtD^{\thresholds}(\fstrats)$ is an NE in $\DA(\dists, \costs)$. Conversely, if $\dstrats$ is an NE in $\DA(\dists, \costs)$, then $\DtF^{(\dists, \thresholds)}(\dstrats)$ is an NE in $\FPA(\dists^{\thresholds})$. 
\end{thm}



Finally, we review a welfare guarantee shown by \citeauthor{KWW16}. 
Combining Theorem~\ref{thm:fpa-pandora-NE-BNE} with known bound on the Price of Anarchy for the first price auction \citep{ST13}, \citeauthor{KWW16} concluded  that the welfare of an NE in a Dutch auction with search costs is at least a $(1-1/e)$-fraction of the maximum expected welfare.\footnote{\cite{jin_first_2022} recently improve the Price of Anarchy bound for the first price auction to $1-1/e^2$.  We do not know whether this bound also holds for the Dutch auction with search costs.}

For our purpose, we generalize their conclusion to $\eps$-NE.  
Formally, let $\varalloc_i(\dstrats)$ be an indicator variable for whether bidder~$i$ receives the item, and let $\varinspect_i(\dstrats)$ be an indicator variable for whether bidder~$i$ inspects her value.
The social welfare of a strategy profile $\dstrats$ is 
\begin{equation}
    \SW^{\DA(\dists, \costs)}(\dstrats)  = \Ex{\sum_{i=1}^n \left(\varalloc_i(\dstrats)\vali - \varinspect_i(\dstrats)\costi \right)},
\end{equation}
where the randomness is over $\vals\sim\dists$ and the randomness of mixed strategies.
Let $\OPT^{(\dists, \costs)}$ be the maximum expected welfare, obtained by Pandora's Box algorithm (Theorem~\ref{thm:optimal-pandora}) on distributions $\disti[1], \ldots, \disti[n]$ and costs $\costi[1], \ldots, \costi[n]$.
\begin{restatable}[Corollary 1 and Theorem 1 of \citealp{KWW16}]{lemma}{optimalwelfare} \label{lem:optimal-welfare}
Let $\thresholdi$ be the index of $(\disti, \costi)$ and $\kappa_i=\min\{\vali, \thresholdi\}$, then 
$\OPT^{(\dists, \costs)} = \Ex{\max_{i\in[n]}\kappa_i}$. 
\end{restatable}
\begin{restatable}[A slight generalization of \citealp{KWW16}]{thm}{epsNEpoa}
\label{thm:epsNEpoa}
Suppose $\dstrats$ is an $\eps$-NE in $\DA(\dists, \costs)$, then $\SW^{\DA(\dists, \costs)}(\dstrats) \ge (1-\frac{1}{e})\OPT^{(\dists, \costs)} - n\eps$. 
\end{restatable}
\noindent The proof of Theorem~\ref{thm:epsNEpoa} follows the smoothness framework \citep{ST13} and is given in Appendix~\ref{sec:proof-thm:epsNEpoa}.

\subsubsection{Transformation with Samples}
\label{sec:sample-KWW16}

We are now ready to present our learning results on auctions with search costs.
In \citet{KWW16}, the utility- and equilibrium-preserving mappings $\FtD^{\thresholds}$ and $\DtF^{(\dists, \thresholds)}$ depend on the value distributions.
We examine the number of samples needed to compute approximations of these mappings, when the value distributions are unknown.
We find that, given search costs and value samples, $\tilde O(1 / \eps^2)$ samples 
suffice to construct mappings between strategies that approximately preserve utility; with $\tilde O(n / \eps^2)$ samples, any equilibrium of the first price auction without search costs on a transformed empirical distribution can be mapped to an approximate equilibrium of the descending auction on the true distribution.
By Theorem~\ref{thm:epsNEpoa}, such an approximate equilibrium in the descending auction must obtain a $(1 - 1/e)$-approximation to the optimal welfare. 
To make use of this result, a market designer could collect $\tilde O(n / \eps^2)$ value samples to compute an approximate Nash in the said FPA, which then maps to an approximate Nash in the Dutch auction.  This approximate Nash can serve as bidding guidance for the participants, and guarantees approximate efficiency of the market.



When value distribution $\disti$'s are unknown (but cost~$\costi$'s are known), the mapping~$\FtD^{\thresholds}$ 
is also unknown, 
because each index~$\thresholdi$ is determined by the distribution~$\disti$. 
Instead, we estimate an index~$\hat{\threshold}_i$ from samples and use the corresponding mapping~$\FtD^{\hat{\thresholds}}$. 

\begin{defn}\label{defn:fpa-pandora-empirical}
Partition the samples $\samples$ into two sets, $\samples^A$ and $\samples^B$, each of size $\nsample/2$. 
Denote the empirical product distributions on $\samples^A$ and $\samples^B$ as $\empDists^A$ and $\empDists$, respectively. 
The \emph{empirical indices} are the indices $\hat{\thresholds}$ for $(\empDists^A, \costs)$; namely, $\hat{\threshold}_i$ is the unique solution to
$\Ex[\vali\sim\empDisti^A]{\max\{\vali - \hat{\threshold}_i, 0\}} = \costi$.
The \emph{empirical counterpart} of $\DA(\dists, \costs)$ is $\FPA(\empDists^{\hat{\thresholds}})$. 
The \emph{empirical mappings} are $\FtD^{\hat{\thresholds}}$ and $\DtF^{(\dists, \hat{\thresholds})}$, computed as in Definition~\ref{def:strategy-mappings}.
\end{defn}

Note that $\DtF^{(\dists, \hat{\thresholds})}$ depends on distributions while $\FtD^{\hat{\thresholds}}$ does not.  The following theorem, analogous to Theorem~\ref{thm:DA_FPA_transform}, shows that the empirical mappings $\FtD^{\hat{\thresholds}}$ and $\DtF^{(\dists, \hat{\thresholds})}$ approximately preserve the utilities with high probability.  
\begin{thm}\label{thm:fpa-pandora-utility-intermediate}
For any $\eps, \delta > 0$, there is 
$M = O \left(\frac{H^2}{\eps^2} \left[\log\left(\frac{H}{\eps} \right) + \log\left(\frac{n}{\delta}\right)\right]\right)$, 
such that for all $\nsample > M$, with probability at least $1-\delta$ over the random draw of $\samples^A$, 
\begin{enumerate}
\item For any monotone mixed strategy profile $\fstrats$ in $\FPA(\dists^{\hat{\thresholds}})$, 
for each bidder~$i$, 
\[ \left| \utili^{\FPA(\dists^{\hat{\thresholds}})}(\fstrats) - \utili^{\DA(\dists, \costs)}(\FtD^{\hat{\thresholds}}(\fstrats)) \right| \le \eps.\]
\item For any mixed strategy profile $\dstrats$ in $\DA(\dists, \costs)$, for each bidder~$i$,
\[ \utili^{\DA(\dists, \costs)}(\dstrats) \le \utili^{\FPA(\dists^{\hat{\thresholds}})}(\DtF^{(\dists, \hat{\thresholds})}(\dstrats)) + \eps.\]
If $\dstrati$ claims above $\hat{\threshold}_i$, then we also have $\utili^{\DA(\dists, \costs)}(\dstrats) \ge \utili^{\FPA(\dists^{\hat{\thresholds}})}(\DtF^{(\dists, \hat{\thresholds})}(\dstrats)) - \eps$. 
\end{enumerate}

\end{thm}
\noindent 
Before proving Theorem~\ref{thm:fpa-pandora-utility-intermediate}, 
we first derive a few important consequences.

\begin{corollary}\label{lem:fpa-pandora-NE-intermediate}
For any $\eps, \delta > 0$, and $\nsample > M$ as in the condition of
Theorem~\ref{thm:fpa-pandora-utility-intermediate}, with probability at least $1-\delta$,  
\begin{enumerate}
\item For any monotone strategy profile $\fstrats$, 
if $\fstrats$ is an $\eps'$-NE in $\FPA(\dists^{\hat{\thresholds}})$, then $\FtD^{\hat{\thresholds}}(\fstrats)$ is an $(\eps'+2\eps)$-NE in $\DA(\dists, \costs)$. 
\item Conversely, for any strategy profile $\dstrats$ that claims above $\hat{\thresholds}$, if $\dstrats$ is an $\eps'$-NE in $\DA(\dists, \costs)$, then $\DtF^{(\dists, \hat{\thresholds})}(\dstrats)$ is an $(\eps'+2\eps)$-NE in $\FPA(\dists^{\hat{\thresholds}})$. 
\end{enumerate}
\end{corollary}

\begin{proof}
We prove the two items respectively, 
\begin{enumerate}
\item Let $\fstrats=(\fstrati, \fstratsmi)$ be an $\eps'$-NE in $\FPA(\dists^{\hat{\thresholds}})$ satisfying the condition in the statement.  For any strategy $\dstrati$, by Theorem~\ref{thm:fpa-pandora-utility-intermediate} item 2, 
\begin{align*}
    \utili^{\DA(\dists, \costs)}(\dstrati, \FtD^{\hat{\thresholds}}(\fstratsmi))
    \le  \utili^{\FPA(\dists^{\hat{\thresholds}})}(\DtF^{(\dists, \hat{\thresholds})}(\dstrati), \DtF^{(\dists, \hat{\thresholds})}(\FtD^{\hat{\thresholds}}(\fstratsmi)))  + \eps. & 
\end{align*}
Since $\fstratsmi$ is monotone, by Claim~\ref{claim:strategy-equivalence} item 2, we have $\DtF^{(\dists, \hat{\thresholds})}(\FtD^{\hat{\thresholds}}(\fstratsmi)) = \fstratsmi$. Thus, 
\begin{align*}
    \utili^{\DA(\dists, \costs)}(\dstrati, \FtD^{\hat{\thresholds}}(\fstratsmi))
    & \le \utili^{\FPA(\dists^{\hat{\thresholds}})}(\DtF^{(\dists, \hat{\thresholds})}(\dstrati), \fstratsmi) + \eps &\\
    \shortintertext{\hfill $\fstrats$ is an $\eps'$-NE in $\FPA(\dists^{\hat{\thresholds}})$}
    & \le \utili^{\FPA(\dists^{\hat{\thresholds}})}(\fstrats) + \eps' + \eps & \\
    \shortintertext{\hfill Theorem \ref{thm:fpa-pandora-utility-intermediate} item 1}
    & \le \utili^{\DA(\dists, \costs)}(\FtD^{\hat{\thresholds}}(\fstrats)) + \eps' + 2\eps. & 
\end{align*}

\item For any strategy $\fstrati$, by Proposition~\ref{prop:monotone}, there exists some monotone strategy $\fstrati'$, such that 
\begin{align*}
    \utili^{\FPA(\dists^{\hat{\thresholds}})}(\fstrati, \DtF^{(\dists, \hat{\thresholds})}(\dstratsmi)) \le \utili^{\FPA(\dists^{\hat{\thresholds}})}(\fstrati', \DtF^{(\dists, \hat{\thresholds})}(\dstratsmi)). 
\end{align*}
Then by Theorem~\ref{thm:fpa-pandora-utility-intermediate} item 1, 
\begin{align*}
    \utili^{\FPA(\dists^{\hat{\thresholds}})}(\fstrati', \DtF^{(\dists, \hat{\thresholds})}(\dstratsmi)) & \le \utili^{\DA(\dists, \costs)}(\FtD^{\hat{\thresholds}}(\fstrati'), \FtD^{\hat{\thresholds}}(\DtF^{(\dists, \hat{\thresholds})}(\dstratsmi))) + \eps.
\end{align*}
Since $\dstratsmi$ claims above $\hat{\thresholds}_{-i}$, by Claim~\ref{claim:strategy-equivalence} item 1, we have $\FtD^{\hat{\thresholds}}(\DtF^{(\dists, \hat{\thresholds})}(\dstratsmi)) = \dstratsmi$. Thus 
\begin{align*}
    \utili^{\FPA(\dists^{\hat{\thresholds}})}(\fstrati, \DtF^{(\dists, \hat{\thresholds})}(\dstratsmi)) & \le \utili^{\DA(\dists, \costs)}(\FtD^{\hat{\thresholds}}(\fstrati'), \dstratsmi) + \eps & \\
    \shortintertext{\hfill $\dstrats$ is an $\eps'$-NE in $\DA(\dists, \costs)$}
    & \le \utili^{\DA(\dists, \costs)}(\dstrats) + \eps' + \eps & \\
    \shortintertext{\hfill Theorem \ref{thm:fpa-pandora-utility-intermediate} item 2}
    & \le \utili^{\FPA(\dists^{\hat{\thresholds}})}(\DtF^{(\dists, \hat{\thresholds})}(\dstrats)) + \eps' + 2\eps. & 
\end{align*}
\end{enumerate}
\end{proof}


As a consequence of Corollary~\ref{lem:fpa-pandora-NE-intermediate},  Corollary~\ref{cor:find-BNE} and Theorem~\ref{thm:epsNEpoa}, any approximate BNE in $\FPA(\empDists^{\hat{\thresholds}})$ is transformed by $\FtD^{\hat{\thresholds}}$ to a nearly efficient approximate NE in $\DA(\dists, \costs)$, as formalized by the following theorem.
\begin{restatable}{thm}{fpapandorane}
\label{thm:fpa-pandora-NE-v2}
For any $\eps, \eps', \delta > 0$, there is $M = O \left(\frac{H^2}{\eps^2} \left[n\log n\log \left(\frac{H}{\eps} \right) + \log\left(\frac{n}{\delta}\right)\right]\right)$, such that for all $\nsample > M$, with probability at least $1-\delta$ over random draws of samples~$\samples$, we have: for
any monotone strategy profile $\fstrats$ that is an $\eps'$-BNE in $\FPA(\empDists^{\hat{\thresholds}})$, $\FtD^{\hat{\thresholds}}(\fstrats)$
is an $(\eps'+4\eps)$-NE in $\DA(\dists, \costs)$; 
moreover, $\SW^{\DA(\dists, \costs)}(\FtD^{\hat{\thresholds}}(\fstrats)) \ge (1-\frac{1}{e})\OPT^{(\dists, \costs)} - n(\eps'+4\eps)$
\end{restatable}
\begin{proof}
First use Corollary~\ref{cor:find-BNE} on distributions $\dists^{\hat{\thresholds}}$.  
Note that $\empDists^{\hat{\thresholds}}$ is an empirical product distribution for $\dists^{\hat{\thresholds}}$; this is because $\empDists$ consists of samples $\samples^B$, whereas $\hat{\thresholds}$ is determined by samples in $\samples^A$, and $\samples^A$ and~$\samples^B$ are disjoint.  Thus, with probability at least $1-\delta/2$ over the random draw of $\samples^B$, any monotone strategy profile $\fstrats$ that is an $\eps'$-BNE in $\FPA(\empDists^{\hat{\thresholds}})$ is an $(\eps'+2\eps)$-BNE in $\FPA(\dists^{\hat{\thresholds}})$.  An $(\eps'+2\eps)$-BNE must be an $(\eps'+2\eps)$-NE in $\FPA(\dists^{\hat{\thresholds}})$, so by Corollary~\ref{lem:fpa-pandora-NE-intermediate}, with probability at least $1-\delta/2$ over the random draw of $\samples^A$, $\FtD^{\hat{\thresholds}}(\fstrats)$ is an $(\eps'+4\eps)$-NE in $\DA(\dists, \costs)$.
The welfare guarantee follows from Theorem~\ref{thm:epsNEpoa}. 
\end{proof}
Theorem \ref{thm:fpa-pandora-NE-v2} does not include the reverse direction, i.e., from an $\eps'$-NE in $\DA(\dists, \costs)$ to an $(\eps'+4\eps)$-BNE in $\FPA(\empDists^{\hat{\thresholds}})$ 
(cf.\@ Theorem~\ref{thm:fpa-pandora-NE-BNE}).  
This is for two reasons:
(1) Such a transformation will result in $(\eps'+4\eps)$-NE in $\FPA(\empDists^{\hat{\thresholds}})$, but $(\eps'+4\eps)$-NE in $\FPA(\empDists^{\hat{\thresholds}})$ is not necessarily an $(\eps'+4\eps)$-BNE.
(2) Unlike interim utility, ex ante utility cannot be learned from samples directly; in other words, $\utili^{\FPA(\empDists^{\hat{\thresholds}})}(\fstrats)$ does not necessarily approximate $\utili^{\FPA(\dists^{\hat{\thresholds}})}(\fstrats)$ even if $\fstrats$ is monotone.  This is because in the computation of ex ante utility we need to take expectation over bidder $i$'s own value but for interim utility we do not need to take such an expectation. 

\paragraph{Proof of Theorem~\ref{thm:fpa-pandora-utility-intermediate}.}
The main idea is as follows: For item 1, we need to show that the utility of a strategy profile $\fstrats$ in $\FPA(\dists^{\hat{\thresholds}})$ approximates the utility of its image $\dstrats=\FtD^{\hat{\thresholds}}(\fstrats)$ in $\DA(\dists, \costs)$. We wish to use Theorem~\ref{thm:DA_FPA_transform} to do so but it cannot be used directly because $\hat{\thresholds}$ is not the indices of $(\dists, \costs)$. Instead, we construct a set of ``empirical costs''~$\hat{\costs}$ such that $\hat{\thresholds}$ becomes the indices of $(\dists, \hat{\costs})$. Then Theorem~\ref{thm:DA_FPA_transform} can be used to show that $\utili^{\FPA(\dists^{\hat{\thresholds}})}(\fstrats)=\utili^{\DA(\dists, \hat{\costs})}(\dstrats)$. With an additional lemma (Lemma~\ref{lem:costs_close}) which shows that $\hat{\costs}$ approximates $\costs$ up to $\eps$-error, we are able to establish the following chain of approximate equations
\[
 \utili^{\FPA(\dists^{\hat{\thresholds}})}(\fstrats) = \utili^{\DA(\dists, \hat{\costs})}(\dstrats) \stackrel{\eps}{\approx} \utili^{\DA(\dists, \costs)}(\dstrats).
\]
The proof for item 2 is similar. 

Formally, define $\hat{\costs}=(\hat{\cost}_i)_{i\in[n]}$, where
\begin{equation}
\hat{\cost}_i \coloneqq \Ex[\vali\sim\disti]{\max\{\vali - \hat{\threshold}_i, 0\}}.
\end{equation}
Note that $\hat{\cost}_i$ is determined by samples~$\samples^A$ since the empirical index $\hat{\threshold}_i$ is computed from $\samples^A$. 
\begin{restatable}{lemma}{costsclose}
\label{lem:costs_close}
There is $M = O\left(\frac{H^2}{\eps^2}\left[\log\frac{H}{\eps} + \log\frac{n}{\delta}\right]\right)$, such that if $m/2 > M$, then with probability at least $1-\delta$ over the random draw of $\samples^A$, for each $i\in[n]$, $|\costi - \hat{\cost}_i|\le\eps$.
\end{restatable}
\begin{proof}
The main idea of the proof is to show that the class $\funcClass_i=\{h^{\threshold}\given\threshold\in[-H, H]\}$ where $h^{\threshold}(x)=\max\{x-r, 0\}$ has pseudo-dimension $\Pdim(\funcClass_i) = O(1)$ and thus uniformly converges with $O\left(\frac{H^2}{\eps^2}\left[\log\frac{H}{\eps} + \log\frac{1}{\delta}\right]\right)$ samples.

Formally, consider the pseudo-dimension $d$ of the class $\funcClass_i=\{h^{\threshold}\given\threshold\in[-H, H]\}$ where $h^\threshold(x)\coloneqq\max\{x-\threshold, 0\}$ for $x\in[0, H]$ (thus $h^\threshold(x)\in[0, 2H]$). We claim that $d=O(1)$. To see this, fix any $d$ samples $(x_1, x_2, \ldots, x_d)$ and any witnesses $(\Pwitnessi[1], \Pwitnessi[2], \ldots, \Pwitnessi[d])$, we bound the number of distinct labelings that can be given by $\funcClass_i$ to these samples. Each sample $x_j$ induces a partition of the parameter space (the space of $\threshold$) $[-H, H]$ into two intervals $[-H, x_j]$ and $(x_j, H]$, such that for any $\threshold\le x_j$, $h^{\threshold}(x_j) = x_j-r$, and for $\threshold > x_j$, $h^{\threshold}(x_j)=0$. All $d$ samples partition $[-H, H]$ into (at most) $d+1$ consecutive intervals, $I_1, \ldots, I_{d+1}$, such that within each interval $I_k$, $h^{\threshold}(x_j)$ is either $x_j-r$ for all $\threshold\in I_k$ or $0$ for all $\threshold\in I_k$, for each $j\in[d]$. We further divide each $I_k$ using witnesses $\Pwitnessi[j]$'s: for each $j\in[d]$, if $h^{\threshold}(x_j) = x_j-r$ for $\threshold\in I_k$, then we cut $I_k$ at the point $r = x_j - \Pwitnessi[j]$; in this way we cut each $I_k$ into at most $d+1$ sub-intervals. Within each sub-interval $I'\subseteq I_k$, the labeling of the $d$ samples given by all $h^{\threshold}$ ($\threshold\in I'$) is the same. Since there are at most $(d+1)^2$ sub-intervals in total, there are at most $(d+1)^2$ distinct labelings.  
To pseudo-shatter $d$ samples, we must have $2^d \leq (d+1)^2$, which gives $d=O(1)$. 

By the definition of $\hat{\threshold}_i$, we have 
\[\costi=\Ex[\vali\sim \empDisti^A]{\max\{\vali - \hat{\threshold}_i, 0\}} = \Ex[\vali\sim \empDisti^A]{\func^{\hat{\threshold}_i}(\vali)}, \]
and $\hat{\threshold}_i\in [-H, H]$. Also note that $\hat{\cost}_i = \Ex[\vali\sim \disti]{\func^{\hat{\threshold}_i}(\vali)}$. 
Thus the conclusion $|\costi-\hat{\cost}_i| \le \eps$ follows from Theorem~\ref{thm:pseudo-dimension} and a union bound over $i\in[n]$.
\end{proof}

\begin{lemma}\label{lem:DA_utility_close}
Suppose $|\costi-\hat{\cost}_i|\le\eps$, then for any strategies~$\dstrats$, 
\begin{equation*}
	\left| \utili^{\DA(\dists, \costs)}(\dstrats) - \utili^{\DA(\dists, \hat{\costs})}(\dstrats)\right|\le\eps.
\end{equation*}
\end{lemma}
\begin{proof}
Couple the realizations of values (and threshold prices and bids if the strategies are randomized) in $\DA(\dists, \costs)$ and $\DA(\dists, \hat{\costs})$. When bidders use the same strategies $\dstrats$ in the two auctions $\DA(\dists, \costs)$ and $\DA(\dists, \hat{\costs})$, bidder~$i$ receives the same allocation and pays the same price.  
The only difference between bidder~$i$'s utilities in these two auctions is the difference between the search costs she pays, which is upper-bounded by $|\costi-\hat{\cost}_i|\le \eps$.
\end{proof}

Now we finish the proof of Theorem~\ref{thm:fpa-pandora-utility-intermediate}. 
\begin{proof}[Proof of Theorem~\ref{thm:fpa-pandora-utility-intermediate}]
First consider item 1.  We use $a\stackrel{\eps}{\approx}b$ to denote $|a-b|\le\eps$.  Given any monotone strategies $\fstrats$ for $\FPA(\empDists^{\hat{\thresholds}})$, 
\begin{align*}
   \utili^{\FPA(\dists^{\hat{\thresholds}})}(\fstrats) 
  ={}& \utili^{\DA(\dists, \hat{\costs})}(\FtD^{\hat{\thresholds}}(\fstrats)) && \text{Theorem \ref{thm:DA_FPA_transform} item 1 } \\
 \stackrel{\eps}{\approx}{}& \utili^{\DA(\dists, \costs)}(\FtD^{\hat{\thresholds}}(\fstrats)) &&\text{Lemma \ref{lem:DA_utility_close}}.  
\end{align*}

Then for item 2, given any strategies $\dstrats$ for $\DA(\dists, \costs)$, by Lemma \ref{lem:DA_utility_close}, 
\begin{align*}
    \utili^{\DA(\dists, \costs)}(\dstrats) 
    \stackrel{\eps}{\approx}  \utili^{\DA(\dists, \hat{\costs})}(\dstrats)
\end{align*}
By Theorem \ref{thm:DA_FPA_transform} item 2, we have  $\utili^{\DA(\dists, \hat{\costs})}(\dstrats) \le \utili^{\FPA(\dists^{\hat{\thresholds}})}(\DtF^{(\dists, \hat{\thresholds})}(\dstrats))$ where ``$=$'' holds if $\dstrati$ claims above $\hat{\threshold}_i$, which concludes the proof. 
\end{proof}

\section{Conclusion}
\label{sec:conclusion}

In this work we obtained almost tight sample complexity bounds for learning utilities in first price auctions and all pay auctions.  
Whereas utilities for unconstrained bidding strategies are hard to learn, we show that learning is made possible by focusing on monotone bidding strategies, which is sufficient for all practical purposes.
We also extended the results to auctions with search costs.

Monotonicity is a natural assumption on bidding strategies in a single item auction, but it does not generalize to multi-parameter settings, where characterization of equilibrium is notoriously difficult.
It is an interesting question whether our results can be generalized to multi-item auctions, such as simultaneous first-price auctions, via more general, lossless structural assumptions on the bidding strategies.  

Our results also depend on the values being drawn independently.  
When bidders' values are correlated, the conditional distribution of opponents' values changes with a bidder's value, and any na\"ive utility learning algorithm needs a number of samples that grows linearly with the size of a bidder's type space.
It is interesting whether there are meaningful tractable middle grounds for utility learning between product distributions and arbitrary correlated distributions.

\bibliographystyle{abbrvnat}
\bibliography{bibs}

\appendix

\section{Missing Proofs from Section \ref{sec:prelim}}
\subsection{Proof of Proposition~\ref{prop:monotone}}
\label{sec:proof-prop:monotone}
\propmonotone*
\begin{proof}
For all practical purposes we may assume $\bidi(\typespacei)$ to be compact.
Fix the distributions $\distsmi$ and strategies $\bidsmi(\cdot)$ of other bidders.
To simplify notation when $\bidsmi(\cdot)$ is fixed, let the interim allocation $\alloci(\bidi)$ be $\Ex[\valsmi \sim \distsmi]{\alloci(\bidi, \bidsmi(\valsmi))}$,  
the interim payment
$\paymenti(\bidi) \coloneqq \Ex[\valsmi \sim \distsmi]{\paymenti(\bidi, \bidsmi(\valsmi))}$,
and the interim utility
$\utili(\vali, \bidi) \coloneqq \utili(\vali, \bidi, \bidsmi(\cdot))$.
Without loss of generality, we may assume for each $\vali$,  $\utili(\vali, \bidi(\vali)) = \max_{\val \in \typespacei} \utili(\vali, \bidi(\val))$
(Otherwise we can first readjust $\bidi(\cdot)$ this way, which only weakly improves the utility of all types.)

Suppose $\bidi(\cdot)$ is non-monotone, i.e., there exist $\vali' > \vali$, such that $\bidi(\vali') < \bidi(\vali)$.  
By the assumption that $\utili(\vali, \bidi(\vali)) = \max_{\val \in \typespacei} \utili(\vali, \bidi(\val))$ for each $\vali$, we have 
\begin{equation}\label{eq:monotone-maximize-1}
    \vali \interimalloc_i(\bidi(\vali)) - \interimpay_i(\bidi(\vali)) \ge \vali \interimalloc_i(\bidi(\vali')) - \interimpay_i(\bidi(\vali'));
\end{equation}
\begin{equation}\label{eq:monotone-maximize-2}
    \vali' \interimalloc_i(\bidi(\vali')) - \interimpay_i(\bidi(\vali')) \ge \vali' \interimalloc_i(\bidi(\vali)) - \interimpay_i(\bidi(\vali)). 
\end{equation}
Adding \eqref{eq:monotone-maximize-1} and \eqref{eq:monotone-maximize-2}, we obtain
\begin{equation}
    (\vali' - \vali)[\interimalloc_i(\bidi(\vali')) - \interimalloc_i(\bidi(\vali))] \ge 0. 
\end{equation}
Since $\vali' > \vali$, we get 
\begin{equation}
    \interimalloc_i(\bidi(\vali')) \ge \interimalloc_i(\bidi(\vali)).
\end{equation}
In both the first price auction and the all pay auction we also have $\interimalloc_i(\bidi(\vali')) \le \interimalloc_i(\bidi(\vali))$ because the probability that $i$ receives the item cannot decrease if her bid increases.  
Therefore, it must be
\begin{equation}\label{eq:monotone-equal-alloc}
    \interimalloc_i(\bidi(\vali')) = \interimalloc_i(\bidi(\vali)). 
\end{equation}

Plugging \eqref{eq:monotone-equal-alloc} into \eqref{eq:monotone-maximize-1} and \eqref{eq:monotone-maximize-2}, we obtain
\begin{equation}\label{eq:monotone-equal-pay}
    \interimpay_i(\bidi(\vali')) = \interimpay_i(\bidi(\vali)). 
\end{equation}


For the all pay auction, since bidder $i$ pays her bid whether or not she wins the item,   \eqref{eq:monotone-equal-pay} implies $\bidi(\vali)=\bidi(\vali')$, a contradiction.  

For the first price auction, for any bid~$\bid$ made by bidder~$i$,  $\interimpay_i(\bid) = \bid\cdot \interimalloc_i(\bid)$.  By \eqref{eq:monotone-equal-pay}, $\bidi(\vali')\interimalloc_i(\bidi(\vali')) = \bidi(\vali) \interimalloc_i(\bidi(\vali))$.  On the other hand, $\interimalloc_i(\bidi(\vali')) = \interimalloc_i(\bidi(\vali))$ and $\bidi(\vali')>\bidi(\vali)$, so we must have
\[\interimalloc_i(\bidi(\vali')) = \interimalloc_i(\bidi(\vali)) = 0. \]
In other words, $\bidi(\vali)$ must be monotone non-decreasing everywhere except maybe for values whose bids are so low that the bidder does not win and hence obtains zero utility.  
Letting the bidder bid~$0$ for all values on which her allocation is~$0$ does not affect her utility and yields a monotone bidding strategy.
\end{proof}

\section{Missing Proofs from Section \ref{sec:fpa}}
\subsection{Upper Bound}
\subsubsection{Proof of Lemma~\ref{lem:pseudo-dimension-utility-class}}
\label{sec:proof-lem:pseudo-dimension-utility-class}
\pdimutil*

\begin{proof}[Proof]
We discussed the case with $n=2$ in Section~\ref{sec:fpa-upper-bound}. Now we consider the general case with $n > 2$ bidders.  
We give the proof for the random-allocation tie-breaking rule;  
the proof for the no-allocation rule is similar (and in fact simpler).
For ease of notation, we use $\Pinputs^k$ to denote $\samplesmi^{k}$.
Recall that each $\Pinputs^k$ is a vector in $\mathbb R^{n-1}$.
We write its $j^{\text{-th}}$ component as $\Pinputi[j]^k$.
We start with a simple observation: for any $\vali$ and $\bids(\cdot)$, the output of $\func^{\vali, \bids(\cdot)}$ on any input 
$\Pinputs^k$
must be one of the following $n+1$ values: 
$\vali-\bidi, \frac{\vali-\bidi}{2}, \ldots, \frac{\vali-\bidi}{n}$, or~$0$; this value is fully determined by the $n-1$ comparisons $\bidi \lesseqqgtr \bidi[j](\Pinputi[j]^k)$ for each $j\ne i$. 
We argue that the hypothesis class $\funcClass_i$ can be divided into $O(\nsample^{2n})$ sub-classes $\{\funcClass_i^{\mathbf k }\}_{\mathbf k \in [m+1]^{2(n-1)}}$ 
such that each sub-class $\funcClass_i^{\mathbf k }$ generates at most $O(\nsample^n)$ different label vectors. 
Thus $\funcClass_i$ generates at most $O(\nsample^{3n})$ label vectors in total. 
To pseudo-shatter $\nsample$~samples, we need $O(\nsample^{3n})\ge 2^\nsample$, which implies $\nsample = O(n\log n)$. 

We now define sub-classes $\{\funcClass_i^{\mathbf k}\}_{\mathbf k}$, each indexed by $\mathbf k \in [m + 1]^{2(n-1)}$.  
For each dimension $j \ne i$, we sort the $\nsample$ samples by their $j^{\text{-th}}$ coordinates non-decreasingly, and use $\pi(j, \cdot)$ to denote the resulting permutation over $\{1, 2, \ldots, \nsample\}$; 
formally, let $\Pinputi[j]^{\pi(j, 1)} \le \Pinputi[j]^{\pi(j, 2)}\le \cdots \le \Pinputi[j]^{\pi(j, \nsample)}$.
For each hypothesis $\func^{\vali, \bids(\cdot)}(\cdot)$, for each $j$, 
we define two special positions; these positions are similar to the position~$k$ in the case for two bidders; 
we now need a pair, because of the need to keep track of ties, due to the more complex random-allocation tie-breaking rule.
Let $k_{j, 1}$ be $\max \{0, \{k: \bidi[j](\Pinputi[j]^{\pi(j, k)}) < \bidi(\vali) \} \}$, and let $k_{j, 2}$ be
$\min \{m + 1, \{k: \bidi[j](\Pinputi[j]^{\pi(j, k)}) > \bidi(\vali) \}\}$.
As in the case for two bidders, this is well defined because of the monotonicity of~$\bidi[j](\cdot)$. It also follows that, if $k_{j, 1} < k_{j, 2} - 1$, then for any $k$ such that $k_{j, 1} < k < k_{j, 2}$, we must have $\bidi[j](\Pinputi[j]^{\pi(j, k)}) = \bidi(\vali)$.
A hypothesis $\func^{\vali, \bids(\cdot)}(\cdot)$ belongs to sub-class $\funcClass_i^{\mathbf k }$ where the index $\mathbf k$ is $(k_{j, 1}, k_{j, 2})_{j\in[n]\backslash\{i\}}$.  
The number of sub-classes is clearly bounded by $(m+1)^{2(n-1)}$.


We now show that the hypotheses within each sub-class $\funcClass_i^{\mathbf k}$ give rise to at most $(m+1)^n$ label vectors.
Let us focus on one such class with index~$\mathbf k$. 
On the $k^{\text{-th}}$ sample~$\Pinputs^k$, 
a hypothesis's membership in  
$\funcClass_i^{\mathbf k}$ suffices to specify whether bidder~$i$ is a winner on this sample, and, if so, the number of other winning bids at a tie.
Therefore, the class index~$\mathbf k$ determines a mapping $c: [m] \to \{0, 1, \ldots, n\}$, with $c(k) > 0$ meaning bidder~$i$ is a winner on sample~$\Pinputs^k$ at a tie with $c(k)-1$ other bidders, and $c(k) = 0$ meaning bidder~$i$ is a loser on sample~$\Pinputs^k$.  
The output of a hypothesis $\func^{\vali, \bids(\cdot)}(\cdot) \in \funcClass_i^{\mathbf k}$ on sample~$\Pinputs^k$ is then $(\vali - \bidi(\vali)) / c(k)$ if $c(k) > 0$ and 0 otherwise.  
The same utility is output on two samples $\Pinputs^k$ and~$\Pinputs^{k'}$ whenever $c(k) = c(k')$.  
Therefore, if we look at the labels assigned to a set~$S$ of samples that are mapped to the nonzero integer by~$c$, there can be at most $|S| + 1 \leq \nsample + 1$ patterns of labels, because we compare the same utility with $|S|$ witnesses; the set of samples mapped to~$0$ by~$c$ have only one pattern of labels. 
The vector of labels generated by a hypothesis in such a sub-class is a concatenation of these patterns.  
The image of $c$ has $n$ nonzero integers, and so there are at most $(\nsample+1)^n$ label vectors.

To conclude, the total number of label vectors generated by $\funcClass_i=\bigcup_{\mathbf k } \funcClass_i^{\mathbf k }$ is at most 
\[ (\nsample+1)^{2(n-1)} (\nsample+1)^{n} \le (\nsample+1)^{3n}. \]
To pseudo-shatter $\nsample$~samples, we need $(\nsample+1)^{3n}\ge 2^\nsample$, which implies $\nsample=O(n\log n)$.

\end{proof}

\subsubsection{Proof of Lemma~\ref{lem:relation-uniform-convergence}}
\label{sec:proof-lem:relation-uniform-convergence}
\uniformprod*

\begin{proof}
Think of the samples $\samples$ as an $\nsample \times n$ matrix $(\samplei^j)$, where each row $j\in[\nsample]$ represents sample~$\samples^j$, and each column~$i\in[n]$ consists of the values sampled from~$\disti$. 
Then we draw $n$ permutations $\pi_1, ..., \pi_n$ of $[\nsample]=\{1, \ldots, \nsample\}$ independently and uniformly at random, and permute the $\nsample$ elements in column~$i$ by~$\pi_i$. 
Regard each new row $j$ as a new sample, denoted by $\permSamples^j = (\samplei[1]^{\pi_1(j)}, \samplei[2]^{\pi_2(j)}, ..., \samplei[n]^{\pi_n(j)})$. 
Given $\pi_1, \ldots, \pi_n$, the ``permuted samples'' $\permSamples^j$, $j=1, \ldots, \nsample$ then have the same distributions as $\nsample$ i.i.d.\@ random draws from~$\dists$. 

For $\func \in \funcClass$, let $p_\func$ be $\Ex[\vals\sim\dists]{h(\vals)}$. 
Then by the definition of $(\eps, \delta)$-uniform convergence (but not on product distribution),
\begin{equation}\label{eq:samples_pi}
\Prx[\samples, \pi]{\exists \func \in \funcClass,\ \left| p_\func - \frac{1}{\nsample }\sum_{j=1}^{\nsample} \func(\permSamples^j) \right|\ge\eps} \le \delta.
\end{equation}

For a set of fixed samples $\samples = (\samples^1, \ldots, \samples^\nsample)$, recall that $\empDisti[i]$ is the uniform distribution over $\{\samplei[i]^{1}, \ldots, \samplei[i]^{\nsample}\}$, and $\empDists = \prod_{i=1}^n \empDisti[i]$. 
We show that the expected 
value of $\func$ on $\empDists$ satisfies $\Ex[\vals\sim\empDists]{\func(\vals)} = \Ex[\pi]{\frac{1}{\nsample}\sum_{j=1}^\nsample \func(\permSamples^j)}$. This is because
\begin{align*}
    \Ex[\pi]{\frac{1}{\nsample}\sum_{i=1}^\nsample \func(\permSamples^j)}
    ={}&\frac{1}{\nsample} \sum_{j=1}^{\nsample} \Ex[\pi]{\func(\permSamples^j)} \\
    ={}&\frac{1}{\nsample}\sum_{j=1}^\nsample \sum_{(k_1, \ldots, k_n)\in[\nsample]^n} \func(\samplei[1]^{k_1}, \ldots, \samplei[n]^{k_n})\ \cdot \\
    & \hspace{8em}
    	\Prx[\pi]{\pi_1(j)=k_1, \ldots, \pi_n(j)=k_n}  \\
    ={}&\frac{1}{\nsample}\sum_{j=1}^\nsample \sum_{(k_1, \ldots, k_n)\in[\nsample]^n} \func(\samplei[1]^{k_1}, \ldots, \samplei[n]^{k_n})\cdot \frac{1}{\nsample^n}\\
    ={}&\frac{1}{\nsample^n} \sum_{(k_1, \ldots, k_n)\in[\nsample]^n} \func(\samplei[1]^{k_1}, \ldots, \samplei[n]^{k_n}) \\
    ={}&\Ex[\vals \sim \empDists]{\func(\vals)}.
\end{align*}

Thus,   
\begin{align*}
     \left| p_\func - \Ex[\vals\sim\empDists]{\func(\vals)}\right| 
     ={}&\left| p_\func - \Ex[\pi]{\frac{1}{\nsample}\sum_{j=1}^\nsample \func(\permSamples^j)} \right| \\
    \le{}& \Ex[\pi]{ \left| p_\func - \frac{1}{\nsample}\sum_{j=1}^\nsample \func(\permSamples^j) \right|}\\
    \le{}& \Prx[\pi]{ \left| p_\func - \frac{1}{\nsample}\sum_{j=1}^\nsample \func(\permSamples^j) \right|\ge \eps}\cdot H \\
    & \hspace{1em} + \left(1-\Prx[\pi]{\left| p_\func - \frac{1}{\nsample}\sum_{j=1}^\nsample \func(\permSamples^j) \right|\ge \eps}\right)\cdot\eps \\
    \le{}& \Prx[\pi]{\mathrm{Bad}(\func, \pi, \samples)}\cdot H + \eps, 
\end{align*}
where in the last step we define event
\[ \mathrm{Bad}(\func, \pi, \samples) = \mathbb{I}\left[\left| p_\func - \frac{1}{\nsample}\sum_{j=1}^\nsample \func(\permSamples^j) \right|\ge \eps\right].\]
By simple calculation,  whenever $\left| p_\func - \Ex[\vals\sim\empDists]{\func(\vals)}\right| \ge 2\eps$, we have $\Prx[\pi]{\mathrm{Bad}(\func, \pi, \samples)}\ge \eps/H$.

Finally, consider the random draw $\samples\sim\dists$, 
\begin{align*}
     \Prx[\samples]{\exists \func\in \funcClass, \ \left| p_\func - \Ex[\vals \sim \empDists]{\func(\vals)} \right|\ge 2\eps} 
    \le{}& \Prx[\samples]{\exists \func \in \funcClass, \ \Prx[\pi]{\mathrm{Bad}(\func, \pi, \samples)}\ge \frac{\eps}{H} } 
    \\
    \le{}& \Prx[\samples]{\Prx[\pi]{\exists \func \in \funcClass, \ \mathrm{Bad}(\func, \pi, \samples)\text{ holds} } \ge \frac{\eps}{H}}.
    \end{align*}
    
    By Markov's inequality, this is in turn upper bounded by
    \begin{align*}
     \frac{H}{\eps}\Ex[\samples]{\Prx[\pi]{\exists \func \in \funcClass, \ \mathrm{Bad}(\func, \pi, \samples)\text{ holds} }  }
    ={}& \frac{H}{\eps}\Prx[\samples, \pi]{\exists \func \in \funcClass, \ \mathrm{Bad}(\func, \pi, \samples)\text{ holds} } \\
    \le{}& \frac{H\delta}{\eps} && \text{ By \eqref{eq:samples_pi}} 
\end{align*}
\end{proof}

\subsection{Lower Bound: Proof of Theorem~\ref{thm:lower-bound-learning-util}}
\label{sec:proof-thm:lower-bound-learning-utility}

\lowerbound*

Fixing $\eps > 0$, fixing $\distconst = 2000$, we first define two value distributions.
Let $\posdist$ be a distribution supported on $\{0, 1\}$, and for $\val \sim \posdist$, $\Prx{\val = 0} = 1 - \frac{1 + \distconst \eps}{n}$, and $\Prx{\val = 1} = \frac{1 + \distconst \eps}{n}$.  
Similarly define $\negdist$: for $\val \sim \negdist$, $\Prx{\val = 0} = 1 - \frac{1 - \distconst \eps}{n}$, and $\Prx{\val = 1} = \frac{1 - \distconst \eps}{n}$.   

Let $\kl(\posdist; \negdist)$ denote the KL-divergence between the two distributions.
\begin{claim}
\label{cl:lb-kl}
$\kl(\posdist; \negdist)= O(\frac {\eps^2}{n})$.
\end{claim}
\begin{proof}
By definition,
\begin{align*}
     \kl(\posdist; \negdist) 
    ={}& \frac{1 + \distconst\eps}{n} \ln \left( \frac{1 + \distconst\eps}{1 - \distconst\eps} \right) 
      + \frac{n - 1 - \distconst \eps}{n} \ln \left(\frac{n - 1 - \distconst \eps}{n - 1 +\distconst \eps}\right) \\
   ={}& \frac 1n  \ln \left( \frac {1 + \distconst\eps}{1 - \distconst \eps} \cdot \frac{(1 - \frac{\distconst \eps}{n - 1})^{n-1}}{(1 + \frac{\distconst \eps}{n-1})^{n-1}}
   \right)
   + \frac {\distconst \eps}{n} \ln \left(\frac{1 + \distconst \eps}{1 - \distconst \eps} \cdot 
   \frac{1 + \frac{\distconst \eps}{n-1}}{1 - \frac{\distconst \eps}{n-1}}\right) \\
   \leq{}& \frac 1n  \ln \left(  \frac {1 + \distconst\eps}{1 - \distconst \eps} \cdot \frac{\left(1 - \frac{\distconst \eps}{n - 1}\right)^{n-1}}{1 + \distconst \eps} \right)
    + \frac {2\distconst \eps}{n} \ln \left(1 + \frac{2\distconst \eps}{1 - \distconst \eps} \right)
\\
   \leq{}& \frac 1 n \ln \left( 
   \frac{1 - \distconst \eps + \frac 1 2 (\distconst \eps)^2}{1 - \distconst \eps}
   \right) + \frac{8\distconst^2 \eps^2}{n} \\
   \leq{}& \frac{10\distconst^2 \eps^2}{n}.
\end{align*}
In the last two inequalities we used $\distconst \eps < \frac 1 2$ and $\ln (1+x) \leq 1+x$ for all $x > 0$.
\end{proof}

It is well known that upper bounds on KL-divergence implies information theoretic lower bound on the number of samples to distinguish distributions \cite[e.g.][]{MansourNotes}.
\begin{corollary}
\label{cor:lb-kl-single}
Given $t$ i.i.d.\@ samples from $\posdist$ or~$\negdist$, if $t \leq \frac{n}{80\distconst^2 \eps^2}$, no algorithm~$\diffalg$ that maps samples to $\{\posdist, \negdist\}$ can do the following: when the samples are from~$\posdist$, $\diffalg$ outputs~$\posdist$ with probability at least $\frac 2 3$, and if the samples are from~$\negdist$, $\diffalg$ outputs~$\negdist$ with probability at least~$\frac 2 3$.  
\end{corollary}

We now construct product distributions using $\posdist$ and~$\negdist$.  
For any $S \subseteq [n - 1]$, define product distribution $\dists_S$ to be $\prod_i \disti$ where $\disti = \posdist$ if $i \in S$, and $\disti = \negdist$ if $i \in [n-1] \setminus S$, and $F_n$ is a point mass on value~$1$.
For any $j \in [n - 1]$ and $S \subseteq [n - 1]$, distinguishing $\dists_{S \cup \{j\}}$ and $\dists_{S \setminus \{j\}}$ by samples from the product distribution is no easier than distinguishing $\posdist$ and $\negdist$, because the coordinates of the samples not from $\disti[j]$ contains no information about~$\disti[j]$.  

\begin{corollary}
\label{cor:lb-kl}
For any $j \in [n - 1]$ and $S \subseteq [n - 1]$, given $t$ i.i.d.\@ samples from $\dists_{S \cup \{j\}}$ or $\dists_{S \setminus \{j\}}$, if $t \leq \frac n {80 \distconst^2 \eps^2}$, no algorithm~$\diffalg$ can do the following: when the samples are from $\dists_{S \cup \{j\}}$, $\diffalg$ outputs~$\dists_{S \cup \{j\}}$ with probability at least $\frac 2 3$, and when the samples are from $\dists_{S \setminus \{j\}}$, $\diffalg$ outputs~$\dists_{S \setminus \{j\}}$ with probability at least~$\frac 2 3$.
\end{corollary}

We now use Corollary~\ref{cor:lb-kl} to derive an information theoretic lower bound on learning utilities for monotone bidding strategies, for distributions in $\{\dists_S\}_{S \subseteq [n]}$.

\begin{proof}[Proof of Theorem~\ref{thm:lower-bound-learning-util}]
Without loss of generality, assume $n$ is odd.  
Let $S$ be an arbitrary subset of~$[n - 1]$ of size either $\lfloor n/2 \rfloor$ or $\lceil n/ 2 \rceil$.
We focus on the interim utility of bidder~$n$ with value~$1$ and bidding $\frac 1 2$.  
Denote this bidding strategy by~$\bidi[n](\cdot)$.
The other bidders may adopt one of two bidding strategies.
One of them is $\posbid(\cdot)$: $\posbid(0) = 0$ and $\posbid(1) = \frac 1 2 + \eta$ for sufficiently small $\eta > 0$.  
The other bidding strategy $\negbid(\cdot)$ maps all values to~$0$.
For $T \subseteq [n-1]$, let $\bids_T(\cdot)$ be the profile of bidding strategies where $\bidi(\cdot) = \posbid(\cdot)$ for $i \in T$, and $\bidi(\cdot) = \negbid(\cdot)$ for $i \notin T$.

For the distribution $\dists_S$, 
\begin{align*}
     \utili[n]\left(1, \frac 1 2, \bids_T(\cdot)\right) 
   ={}& \frac 1 2  \Prx{\max_{i \in T} \vali = 0} \\
    ={}&  \frac 1 2 
    \left(1 - 
    \frac{1 + \distconst \eps}{n}
    \right)^{|S \cap T|}
    \left(1 - 
    \frac{1 - \distconst \eps}{n}
    \right)^{|T\setminus S|}
    \\
    ={}& \frac 1 2
    \left(
    1 - \frac{1 + \distconst \eps}{n}
    \right)^{|T|}
    \left(
    \frac{n - 1 + \distconst \eps}{n - 1 - \distconst \eps}
    \right)^{|T \setminus S|}.
\end{align*}
Therefore, for $T, T' \subseteq [n-1]$ with $|T| = |T'| $,
\begin{align*}
     \frac{\utili[n](1, \frac 1 2, \bids_T(\cdot))}{\utili[n](1, \frac 1 2, \bids_{T'}(\cdot))} ={}& \left(
    1 + \frac{2\distconst \eps / (n-1)}{1 - \frac{\distconst \eps}{n-1}} 
    \right)^{|T\setminus S| - |T' \setminus S|} \\
     \geq{}& 1 + \frac{2\distconst \eps}{n-1} \cdot (|T \setminus S| - |T' \setminus S|);
\end{align*}
Suppose $|T \setminus S| \geq |T' \setminus S|$ and $|T| = |T'| \geq \lfloor \frac n 2 \rfloor$, then
\begin{align}
 \utili[n]\left(1, \frac 1 2, \bids_T(\cdot)\right) - \utili[n]\left(1, \frac 1 2, \bids_{T'}(\cdot)\right) \ge{}& (|T \setminus S| - |T' \setminus S|) \cdot \frac {2\distconst \eps}{n-1} \cdot \utili[n]\left(1, \frac 1 2, \bids_{T'}(\cdot) \right) 
\notag \\
\geq{}& (|T \setminus S| - |T' \setminus S|) \cdot \frac {2\distconst \eps}{n-1} \cdot \frac 1 {8 e^2}, 
\label{eq:util-diff-eps}
\end{align}
where the last inequality is because $\utili[n](1, \frac 1 2, \bids_{T'}(\cdot)) \ge \frac{1}{2} (1 - \frac{2}{n})^n = \frac{1}{2} [(1 - \frac{2}{n})^\frac{n}{2}]^2\ge \frac{1}{2}  (\frac{1}{2e})^2 = \frac{1}{8e^2}$. 

Now suppose an algorithm~$\alg$ $(\eps, \delta)$-learns the utilities of all monotone bidding strategies with $t$ samples~$\samples$ for $t \leq \frac{n}{80\distconst^2 \eps^2}$.
Define $\diffalg: \mathbb R_+^{n \times t} \times \mathbb N \to 2^{[n-1]}$ be a function that outputs among all $T\subseteq [n-1]$ of size~$k$, the one that maximizes bidder~$n$'s utility when they bid according to bidding strategy $\bids_T$.  
Formally, 
\begin{align*}
    \diffalg(\samples, k) = \argmax_{T \subseteq [n-1], |T| = k} \alg\left(\samples, n, 1, (\bids_{T}(\cdot), \bidi[n](\cdot)) \right),
\end{align*}

By Definition~\ref{def:util-learn-ensemble}, for any $S$ with $|S| = \lfloor n / 2 \rfloor$, for samples drawn from~$\dists_S$, with probability at least $1 - \delta$,
\begin{equation*}
    \alg(\samples, n, 1, (\bids_{[n-1]\setminus S}(\cdot), \bidi[n](\cdot)) \\
    \geq \utili[n]\left(1, \frac 1 2, \bids_{[n-1] \setminus S}(\cdot) \right) - \eps;
\end{equation*}
and for any $T \subseteq[n-1]$ with $|T| = \lceil n / 2 \rceil$,
\begin{equation*}
    \alg(\samples, n, 1, (\bids_T(\cdot), \bidi[n](\cdot))
    \leq \utili[n]\left(1, \frac 1 2, \bids_T(\cdot) \right) + \eps.
\end{equation*}
Therefore, for $W = \diffalg(\samples, \lceil n / 2 \rceil)$, 
\begin{align*}
    \utili[n]\left(1, \frac 1 2, \bids_W(\cdot) \right) \geq 
    \utili[n] \left(1, \frac 1 2, \bids_{[n-1]\setminus S}(\cdot) \right) - 2\eps.
\end{align*}
Since $|W| = [n-1]\setminus S = \lceil n / 2 \rceil$, by \eqref{eq:util-diff-eps},
\begin{align*}
    \left(\lceil \frac n 2 \rceil - |W \setminus S| \right) \cdot \frac{\distconst \eps}{(n-1)4e^2} \leq 2\eps.
\end{align*}
So
\begin{align*}
    |W \cap S| \leq (n - 1) \cdot \frac{8e^2}{\distconst}.
\end{align*}
In other words, with probability at least $ 1- \delta$, $\diffalg(\samples, \lceil n / 2 \rceil)$ is the complement of~$S$ except for at most $\frac{8e^2}{\distconst}$ fraction of the coordinates in $[n-1]$.

Similarly, for $S$ of cardinality $\lceil n / 2 \rceil$, 
\begin{align*}
    |\diffalg(\samples, \lceil n / 2 \rceil) \cap S| \leq (n - 1) \cdot \frac{8e^2}{\distconst} + 1.
\end{align*}
Take $\diffconst$ to be $\frac{8e^2}{\distconst}$. We have $\diffconst<\frac 1 {20}$. For all large enough $n$ and all $S$ of size~$\lfloor n / 2 \rfloor$ or $\lceil n / 2 \rceil$, with probability at least $1 - \delta$, $\diffalg(\samples, \lceil n / 2 \rceil)$ correctly outputs the elements not in~$S$ with an exception of at most $\diffconst$ fraction of coordinates.

Let $\coordsets$ be the set of all subsets of $[n-1]$ of size either $\lceil n / 2 \rceil$ or $\lfloor n / 2 \rfloor$.
Consider any~$S \in \coordsets$.  
Let $\diffset(S) \subseteq [n-1]$ denote the set of coordinates whose memberships in~$S$ are correctly predicted by $\diffalg(\samples, \lceil n / 2 \rceil)$ with probability at least $2/3$; that is, $i \in \diffset(S)$ if{f} with probability at least $2/3$, $\diffalg(\samples, \lceil n / 2 \rceil)$ is correct about whether $i \in S$.  
Let the cardinality of~$|\diffset(S)|$ be $z(n-1)$. Suppose we draw coordinate $i$ uniformly at random from $[n-1]$, and independently draw samples $\samples$ from $\dists_S$, then the probability that $\diffalg(\samples, \lceil n / 2 \rceil)$ is correct about whether $i\in S$ satisfies:
\begin{align*}
     \Prx[i, \samples]{\diffalg(\samples, \lceil n / 2 \rceil)\text{ is correct about whether }i\in S}
     \geq{}& (1 - \diffconst) (1 - \delta) \\
     \geq{}& 0.9,
\end{align*}
and 
\begin{align*}
    \Prx[i, \samples]{\diffalg(\samples, \lceil n / 2 \rceil)\text{ is correct about whether }i\in S}
    \le{}& \Prx[i]{i\in \diffset(S)}\cdot 1 + \Prx[i]{i\notin \diffset(S)}\cdot\frac{2}{3}  \\
    ={}& z\cdot 1 + (1-z)\cdot \frac 2 3,
\end{align*} 
which implies $z > 0.6$.  
If a pair of sets $S$ and~$S'$ differ in only one coordinate~$i$, and $i \in \diffset(S) \cap \diffset(S')$, then $\diffalg(\cdot)$ serves as an algorithm that tells apart $\dists_S$ and~$\dists_{S'}$, contradicting Corollary~\ref{cor:lb-kl}.  
We now show, with a counting argument, that such a pair of $S$ and~$S'$ must exist.

Since for each $S \in \coordsets$, $|\diffset(S)| \geq 0.6(n-1)$, there exists a coordinate $i \in [n-1]$ and $\mathcal T \subseteq \coordsets$, with $|\mathcal T| \geq 0.6 |\coordsets|$, such that for each $S \in \mathcal T$, $i \in \diffset(S)$. 
But $\coordsets$ can be decomposed into $|\coordsets| / 2$ pairs of sets, such that within each pair, the two sets differ by one in size, and precisely one of them contains coordinate~$i$.  
Therefore among these pairs there must exist one $(S, S')$ with $S, S' \in \mathcal T$, i.e., $i \in \diffset(S)$ and $i \in \diffset(S')$.
Using $\diffalg$, which is induced by~$\alg$, we can tell apart $\dists_S$ and~$\dists_{S'}$ with probability at least $2/3$, which is a contradiction to Corollary~\ref{cor:lb-kl}.
This completes the proof of Theorem~\ref{thm:lower-bound-learning-util}.
\end{proof}

\section{Missing Proofs from Section \ref{sec:search}}
\subsection{Proof of Theorem \ref{thm:epsNEpoa}}
\label{sec:proof-thm:epsNEpoa}
\epsNEpoa*

Recall that $\varalloc_i(\dstrats)$ indicates whether bidder~$i$ receives the item, and 
$\varinspect_i(\dstrats)$ indicates whether bidder~$i$ inspects her value.
The expected utility of bidder~$i$ can be decomposed into a welfare term and a payment term, as follows:
\[\utili^{\DA(\dists, \costs)}(\dstrats) = \Ex{\varalloc_i(\dstrats)(\vali - \dbidi(\vali)) - \varinspect_i(\dstrats)\costi} = \Ex{\varalloc_i(\dstrats)\vali - \varinspect_i(\dstrats)\costi} - \Ex{\varalloc_i(\dstrats)\dbidi(\vali)},\]
where the randomness is over $\vals\sim\dists$ and the randomness of mixed strategies.
And the social welfare can be expressed as the sum of utilities and payments of bidders: 
\begin{align*}
    \SW^{\DA(\dists, \costs)}(\dstrats)  = \sum_{i=1}^n\Ex{ \varalloc_i(\dstrats)\vali - \varinspect_i(\dstrats)\costi}
    = \sum_{i=1}^n \left[ \utili^{\DA(\dists, \costs)}(\dstrats) + \Ex{\varalloc_i(\dstrats)\dbidi(\vali)}\right].
\end{align*}

Now suppose $\dstrats$ is an $\eps$-NE, then 
\begin{align}
    \SW^{\DA(\dists, \costs)}(\dstrats) & \ge \sum_{i=1}^n \left[\utili^{\DA(\dists, \costs)}(\dstrati', \dstratsmi)-\epsilon + \Ex{\varalloc_i(\dstrats)\dbidi(\vali)}\right] \nonumber \\
    & = \sum_{i=1}^n \utili^{\DA(\dists, \costs)}(\dstrati', \dstratsmi) + \sum_{i=1}^n \Ex{\varalloc_i(\dstrats)\dbidi(\vali)} - n\eps, 
\label{eq:SW-eps-NE}
\end{align}
for any set of strategies $\{\dstrati'\}_{i\in[n]}$.

For each $i$, let $\thresholdi$ be the index of $(\disti, \costi)$. 
Recall that we use $\kappa_i$ to denote $\min\{\vali, \thresholdi\}$. 
We will construct strategies $\dstrati'$ that satisfy the following inequality
\begin{equation}
\label{eq:smoothness-goal}
    \sum_{i=1}^n \utili^{\DA(\dists, \costs)}(\dstrati', \dstratsmi) \ge (1-\frac{1}{e})\Ex{\max_{i\in[n]} \kappa_i} - \sum_{i=1}^n \Ex{\varalloc_i(\dstrats)\dbidi(\vali)}. 
\end{equation}
By \eqref{eq:SW-eps-NE} and \eqref{eq:smoothness-goal} we have
$\SW^{\DA(\dists, \costs)}(\dstrats) \ge (1-\frac{1}{e})\Ex{\max_{i\in[n]} \kappa_i} - n\eps$. Since $\Ex{\max_{i\in[n]} \kappa_i} = \OPT^{(\dists, \costs)}$ (Lemma~\ref{lem:optimal-welfare}), the theorem is proved. 

Now we construct $\dstrati'$.  Each $\dstrati'$ is a mixed strategy that does the following: sample a random variable $Z\in [\frac{1}{e}, 1]$ with probability density function $f_Z(z) = \frac{1}{z}$; inspect the value at threshold price $\dtimei'=(1-Z)\thresholdi$; and claim the item at price $\dbidi'(\vali) = (1-Z)\kappa_i$.  Note that $\dstrati'$ claims above $\thresholdi$, and we will make use of the following property of strategies that claim above $\thresholdi$: 
\begin{claim}\label{cl:utility-pandora}
For strategy $\dstrati'$ that claims above $\thresholdi$, 
$\Ex{\varalloc_i(\dstrat', \dstratsmi)\vali - \varinspect_i(\dstrati', \dstratsmi)\costi} = \Ex{\varalloc_i(\dstrati', \dstratsmi)\kappa_i}$.
\end{claim}
\begin{proof}
For convenience we write $\varalloc_i = \varalloc_i(\dstrati', \dstratsmi)$ and $\varinspect_i = \varinspect_i(\dstrati', \dstratsmi)$.  By linearity of expectation and the definition of index,
\begin{align*}
    \Ex{\varalloc_i\vali - \varinspect_i\costi}
     & = \Ex{\varalloc_i\vali}  - \Ex{\varinspect_i}\costi \\
     & = \Ex{\varalloc_i\vali}  - \Ex{\varinspect_i}\Ex[\vali\sim\disti]{\max\{\vali - \thresholdi, 0\}}. 
\end{align*}
Note that $\vali$ and $\varinspect_i$ are independent because bidder~$i$ doesn't know her value before inspection. Thus
\begin{align*}
    \Ex{\varalloc_i\vali - \varinspect_i\costi}
    & = \Ex{\varalloc_i\vali}  - \Ex{\varinspect_i\max\{\vali - \thresholdi, 0\}} \\
    & = \Ex{\varalloc_i\vali - \varinspect_i \max\{\vali - \thresholdi, 0\}} \\
    & = \Ex{\varalloc_i\vali - \varalloc_i \max\{\vali - \thresholdi, 0\} +  (\varalloc_i - \varinspect_i) \max\{\vali - \thresholdi, 0\}}. 
\end{align*}
Because $\dstrati'$ claims above $\thresholdi$, we have $\varalloc_i = \varinspect_i$ whenever $\vali>\thresholdi$.  This implies $(\varalloc_i - \varinspect_i) \max\{\vali - \thresholdi, 0\} = 0$ and 
\begin{align*}
    \Ex{\varalloc_i\vali - \varinspect_i\costi}
     = \Ex{\varalloc_i(\vali - \max\{\vali - \thresholdi, 0\})}
     = \Ex{\varalloc_i\kappa_i}.
\end{align*}
\end{proof}

Now we argue that the $\{\dstrati'\}_{i\in[n]}$ constructed above satisfy \eqref{eq:smoothness-goal}. By Claim~\ref{cl:utility-pandora}, we have $\utili^{\DA(\dists, \costs)}(\dstrati', \dstratsmi) = \Ex{\varalloc_i(\dstrati', \dstratsmi)(\kappa_i - \dbidi'(\vali))}$.  Summing over $i\in[n]$,  
\begin{align*}
    \sum_{i=1}^n \utili^{\DA(\dists, \costs)}(\dstrati', \dstratsmi) & = \sum_{i=1}^n \Ex{\varalloc_i(\dstrati', \dstratsmi)(\kappa_i - \dbidi'(\vali))} \\ 
    & = \Ex{ \sum_{i=1}^n \varalloc_i(\dstrati', \dstratsmi)(\kappa_i - \dbidi'(\vali))}  \\
    & = \Ex{ \sum_{i=1}^n \varalloc_i(\dstrati', \dstratsmi) Z\kappa_i}.
\end{align*}
For any fixed value profile $\vals=(\vali)$, let $i^*\coloneqq \argmax_{i\in[n]}\{\kappa_i\}$.  Since  $\varalloc_i(\dstrati', \dstratsmi) Z\kappa_i \ge 0$, we have 
\begin{align}\label{eq:utility-maxindex}
    \sum_{i=1}^n \utili^{\DA(\dists, \costs)}(\dstrati', \dstratsmi) \ge \Ex{ \varalloc_{i^*}(\dstrati[i^*]', \dstratsmi[i^*]) Z\kappa_{i^*}}.
\end{align}
\begin{claim}\label{cl:smoothness-step}
For any $\vals$, $\Ex{ \varalloc_{i^*}(\dstrati[i^*]', \dstratsmi[i^*]) Z\kappa_{i^*}\mid \vals} \ge (1-\frac{1}{e}) \kappa_{i^*} - \sum_{i=1}^n \varalloc_i(\dstrats)\dbidi(\vali)$. 
\end{claim}
\begin{proof}
Let $p\coloneqq \max_{j\ne {i^*}} \dbidi[j](\vali[j])$.
If $p\ge (1-\frac{1}{e})\kappa_{i^*}$, then $\Ex{ \varalloc_{i^*}(\dstrati[i^*]', \dstratsmi[i^*]) Z\kappa_{i^*}\mid \vals} \ge 0 \ge (1-\frac{1}{e}) \kappa_{i^*} - p$.
Otherwise, note that whenever bidder~$i^*$'s bid $(1-Z)\kappa_{i^*}$ is above $p$, she wins the item, thus
\begin{align*}
    \Ex{ \varalloc_{i^*}(\dstrati[i^*]', \dstratsmi[i^*]) Z\kappa_{i^*}\mid \vals} & = \int_{1/e}^{1-p/\kappa_{i^*}} z\kappa_{i^*} f_Z(z) \dd z  = \int_{1/e}^{1-p/\kappa_{i^*}} z\kappa_{i^*} \frac{1}{z} \dd z \\
    & = (1- \frac{1}{e} - \frac{p}{\kappa_{i^*}}) \kappa_{i^*}  = (1- \frac{1}{e})\kappa_{i^*}  - p.
\end{align*}
The proof is completed by observing that $\sum_{i=1}^n \varalloc_i(\dstrats)\dbidi(\vali) = \max_{i\in[n]} \bidi(\vali) \ge p$.
\end{proof}
Taking expectation over $\vals\sim\dists$, \eqref{eq:utility-maxindex} and Claim~\ref{cl:smoothness-step} immediately implies \eqref{eq:smoothness-goal}.

\end{document}